\documentclass[a4paper,11pt]{article}
 \usepackage[]{geometry}
 \usepackage[utf8]{inputenc}
 \usepackage[parfill]{parskip}
 \usepackage{graphicx}	
 \usepackage{caption}
 \usepackage{subcaption}		
 \usepackage{amssymb}
 \usepackage{amsmath}
 \usepackage{amsthm}
 \usepackage{natbib}
 \usepackage[flushleft]{threeparttable}
 \usepackage{tabularx}
 \usepackage{booktabs}
 \usepackage{bbm}
 \usepackage{xcolor}
 \usepackage[title]{appendix}
 \usepackage[hyphens]{url} 
 \usepackage[colorlinks,citecolor=blue,urlcolor=blue,bookmarks=false,hypertexnames=true]{hyperref}
 \usepackage{tikz}
 \usetikzlibrary{arrows}
 \usepackage{enumitem}

 \newcommand\numberthis{\addtocounter{equation}{1}\tag{\theequation}}

 \newtheorem{assume}{Assumption}
 \newtheorem{theorem}{Theorem} 
 \newtheorem{remark}{Remark}

 \title{The Synthetic Control Method with Nonlinear Outcomes: Estimating the Impact of the 2019 Anti-Extradition Law Amendments Bill Protests on Hong Kong's Economy}

 \author{Wei Tian}
\date{August 24, 2021}

\begin{document}

 \maketitle

 \begin{abstract}
 	The synthetic control estimator \citep{abadie2010synthetic} is asymptotically unbiased assuming that the outcome is a linear function of the underlying predictors and that the treated unit can be well approximated by the synthetic control before the treatment. When the outcome is nonlinear, the bias of the synthetic control estimator can be severe.
 	In this paper, we provide conditions for the synthetic control estimator to be asymptotically unbiased when the outcome is nonlinear, and propose a flexible and data-driven method to choose the synthetic control weights.
 	Monte Carlo simulations show that compared with the competing methods, the nonlinear synthetic control method has similar or better performance when the outcome is linear, and better performance when the outcome is nonlinear, and that the confidence intervals have good coverage probabilities across settings.
 	In the empirical application, we illustrate the method by estimating the impact of the 2019 anti-extradition law amendments bill protests on Hong Kong's economy, and find that the year-long protests reduced real GDP per capita in Hong Kong by 11.27\% in the first quarter of 2020, which was larger in magnitude than the economic decline during the 1997 Asian financial crisis or the 2008 global financial crisis.
 \end{abstract}

\section{Introduction}

The synthetic control method \citep{abadie2010synthetic} estimates the treatment effect on a treated unit by comparing its outcome with the outcome of the synthetic control, which is constructed using a convex combination of the control units such that the observed predictors and pretreatment outcomes of the synthetic control closely match those of the treated unit.
\cite{abadie2010synthetic} show that the bias of the synthetic control estimator is bounded by a function that goes to zero as the number of pretreatment periods increases, provided that the outcome is a linear function of the observed and unobserved predictors, and that the treated unit is well approximated by the synthetic control in the pretreatment periods.
When the outcome is nonlinear, however, the bias of the synthetic control estimator can be severe, as a good fit on the observed predictors and pretreatment outcomes between the treated unit and the synthetic control does not necessarily imply a good fit on the unobserved predictors.
This paper generalises the synthetic control method to cases where the outcome is a nonlinear function of the predictors.

We first relax the non-negativity restriction on the weights, which is imposed by \cite{abadie2010synthetic} to prevent extrapolation biases.
We show using an example that there is no extrapolation bias when the outcome is linear, and when the outcome is nonlinear, the interpolation bias or extrapolation bias tends to be smaller if we construct the synthetic control using control units that are closer to the treated unit, whereas the synthetic control method with the non-negativity restriction does not prefer or implement the use of closer neighbours.
The non-negativity restriction also makes it less likely to obtain a good pretreatment fit, especially when the treated unit takes extreme values in the matching variables or when the sample size is small, and thus limits the applicability of the synthetic control method. In cases where the synthetic control method can be used, the non-negativity restriction may distort the size of the permutation test since the post/pretreatment RMSPE ratio in the permutation test is conditional on a good pretreatment fit for the treated unit but unconditional for the others, which would lead to over-rejection of the null hypothesis, as noted in \cite{ferman2017placebo}.
Although the non-negativity restriction acts as a regularisation method and often ensures the sparsity of the weights, which is an appealing feature in comparative case studies due to the ease of interpretation, this may come at the cost of a larger bias of the estimator. Relaxing the non-negativity restriction allows more flexible regularisation methods to be implemented so that the bias (and potentially the variance) of the estimator is smaller, while the sparsity of the weights can also be achieved when appropriate.

We then move on to provide the conditions for the synthetic control estimator to be asymptotically unbiased when the outcome is nonlinear, to complement the theoretical result for the linear case in \cite{abadie2010synthetic}. 
Furthermore, we show that there is a tradeoff between the aggregate matching discrepancy, i.e., the matching discrepancy between the treated unit and the synthetic control, and the pairwise matching discrepancies, i.e., the matching discrepancies between the treated unit and the control units used for constructing the synthetic control, depending on the degree of nonlinearity. Specifically, when the degree of nonlinearity is low, the bias of the synthetic control estimator tends to be smaller if we construct the synthetic control using more control units, so that the matching discrepancy between the treated unit and the synthetic control is smaller and the weights are more spread out. When the outcome is highly nonlinear, the bias of the synthetic control estimator tends to be smaller if we use only the nearest neighbours, so that the pairwise matching discrepancies are smaller.

To address this trade-off, we propose choosing the weights with elastic net type regularisation, where the $L_1$ penalty terms are weighted by pairwise matching discrepancies between the treated unit and the control units to penalise using control units that are farther away from the treated unit, whereas the $L_2$ penalty term penalises concentrating the weights on a few control units, and the optimal tuning parameters for the penalty terms are selected using cross-validation.
This method can be considered as a combination of the methods in \cite{doudchenko2017} and \cite{abadie2020penalized}, where \cite{doudchenko2017} propose choosing the weights with the elastic net regularisation while relaxing the non-negativity restriction and other restrictions on the weights, and \cite{abadie2020penalized} propose choosing the weights with $L_1$ penalty terms weighted by pairwise matching discrepancies between the treated unit and the control units while maintaining the non-negativity restriction.
The motivation of both these studies is to use regularisation methods to ensure that there is a unique set of weights that minimise the matching discrepancy between the treated unit and the synthetic control when the number of control units is large with no regard to nonlinearity, whereas this paper aims to provide a flexible and data-driven method to obtain a synthetic control estimator that has a smaller bias when the outcome is potentially nonlinear.
Monte Carlo simulations comparing the nonlinear synthetic control method and these two methods as well as the original synthetic control method show that the nonlinear synthetic control method has similar performance with the method in \cite{doudchenko2017} and better performance than the other two in linear settings, and has the best performance in nonlinear settings.

In the main empirical application, we estimate the impact of the 2019 anti-extradition law amendments bill protests on Hong Kong's economy. The results suggest that the protests had a detrimental effect on Hong Kong's economy from the second quarter of 2019. The magnitude of the impact grew rapidly and reached its peak in the first quarter of 2020, when real GDP per capita in Hong Kong was 11.27\% lower than what it would be if there were no protests. This exceeds the peak-to-trough decline in quarterly real GDP per capita in Hong Kong during the 1997 Asian financial crisis and the 2008 global financial crisis.
The effect became insignificant in the second and third quarters of 2020, when almost all economies were severely hit by the COVID-19 pandemic, and was significant again in the fourth quarter, with the quarterly GDP per capita 8.8\% lower than its counterfactual level due to the slow recovery of the economy in Hong Kong.

The rest of the paper is organised as follows.
Section \ref{Ch1_sec_SCM} provides an overview of the original synthetic control method and a discussion on the non-negativity restriction. Section \ref{Ch1_sec_NSC} provides the conditions for the synthetic control estimator to be asymptotically unbiased when the outcome is nonlinear, and proposes the nonlinear synthetic control method for choosing the weights. Section \ref{Ch1_sec_sim} conducts Monte Carlo simulations to compare the nonlinear synthetic control method and the competing methods. Section \ref{Ch1_sec_app} revisits the two applications in \cite{abadie2010synthetic} and \cite{abadie2015comparative} to illustrate the nonlinear synthetic control method, and estimates the impact of the 2019 anti-extradition law amendments bill protests on Hong Kong's economy. Section \ref{Ch1_sec_con} concludes. Appendix \ref{Ch1_sec_data} lists the data sources for the main application. Appendix \ref{Ch1_sec_proof} collects the proofs.

\section{The Synthetic Control Method}\label{Ch1_sec_SCM}

\subsection{Overview}
This section provides an overview of the original synthetic control method in \cite{abadie2010synthetic}. For a more detailed review, see \cite{abadie2021jel}.
Suppose that we observe $N$ units over $T$ time periods. Without loss of generality, we assume that the first unit receives treatment at period $T_0+1\le T$ and remains treated afterwards, while all the other $J=N-1$ units are untreated throughout the window of observation.
If we denote the indicator for the binary treatment status for unit $i$ at time $t$ as $D_{it}$, then $D_{it}=1$ for $i=1$ and $t>T_0$, and $D_{it}=0$ otherwise.

The quantity of interest is the treatment effect on the treated unit at time $t>T_0$, which is given by the difference between its treated potential outcome and untreated potential outcome at time $t$ \citep{rubin1974estimating},
$$\tau_{1t}=Y_{1t}^1-Y_{1t}^0,$$
where $Y_{1t}^1$ is the outcome that we would observe for unit 1 at time $t$ if unit 1 is treated at the time, and $Y_{1t}^0$ is the outcome that we would observe otherwise.
The observed outcome can be written as $Y_{1t}=D_{1t}Y_{1t}^1+\left(1-D_{1t}\right)Y_{1t}^0$.
Since we only observe $Y_{1t}^1$ for $t>T_0$, estimating $\tau_{1t}$ requires predicting the untreated potential outcome $Y_{1t}^0$. We assume the following functional form for the untreated potential outcome.

\begin{assume}[]\label{Ch1_assume_IFE}
	The untreated potential outcome for unit $i$ at period $t$ is given by an interactive fixed effects model
	\begin{equation}\label{Ch1_eq_IFE}
		Y_{it}^0=\boldsymbol{X}_i'\boldsymbol{\beta}_t+\boldsymbol{\mu}_i'\boldsymbol{\lambda}_t+\varepsilon_{it},
	\end{equation}
	where $\boldsymbol{X}_i$ and $\boldsymbol{\mu}_i$ are the $k\times 1$ and $f\times 1$ vectors of observed and unobserved predictors of $Y_{it}^0$ with coefficients $\boldsymbol{\beta}_{t}$ and $\boldsymbol{\lambda}_{t}$, respectively, and $\varepsilon_{it}$ is the individual transitory shock.
\end{assume}

\begin{remark}
	\normalfont
	Alternatively, the interactive fixed effects term $\boldsymbol{\mu}_i'\boldsymbol{\lambda}_t$ can be interpreted as the product of common time factors $\boldsymbol{\lambda}_t$ and individual factor loadings $\boldsymbol{\mu}_i$. Both interpretations are discussed in more details in \cite{bai2009panel}.
\end{remark}

\begin{remark}
	\normalfont
	The factor model presented in \cite{abadie2010synthetic} and \cite{abadie2021jel} also includes a common time-varying intercept, representing the time trend in the outcome. This can be considered as a special case of \eqref{Ch1_eq_IFE}, where the time-varying intercept is one element in $\boldsymbol{\lambda}_t$ with the corresponding element in $\boldsymbol{\mu}_i$ being 1.
\end{remark}

\medskip

The individual transitory shocks are assumed to satisfy the following assumptions.

\begin{assume}[]\label{Ch1_assume_error}
	\leavevmode
	\begin{enumerate}[label=\arabic*)]
		\item $\varepsilon_{it}$ are independent across $i$ and $t$;
		\item $\mathbbm{E}\left(\varepsilon_{it}\mid \boldsymbol{X}_{j},\boldsymbol{\mu}_j,D_{js}\right)=0$ for all $i$, $j$, $t$ and $s$;
		\item $\mathbb{E}\vert\varepsilon_{it}\vert^p<\infty$ for all $i$, $t$ and some even integer $p\ge2$.
	\end{enumerate}
\end{assume}

\begin{remark}
	\normalfont
	The first part of Assumption \ref{Ch1_assume_error} assumes that the individual transitory shocks are independent across units and time.
	In a panel data setting, cross sectional and time serial correlations in the individual transitory shocks are to be expected. Here we are making a simplifying assumption that the cross sectional and time serial correlations are due to the unobserved individual and time fixed effects. Indeed, if we treat $u_{it}=\boldsymbol{\mu}_i'\boldsymbol{\lambda}_t+\varepsilon_{it}$ as the individual transitory shock, then $u_{it}$ are correlated across units and time, while $\varepsilon_{it}$ remain independent across units and time.
	The second part assumes that the individual transitory shocks have zero mean conditional on the observed and unobserved predictors and the treatment status.
	The third part ensures that the predictors are not dominated by the transitory shocks in determining the outcomes.
\end{remark}

To estimate the treatment effect for unit 1 at time $t>T_0$, a synthetic control is constructed as a linear combination of the control units using weights $w_j,\ j=2,\dots,N$ such that
\begin{align*}
	\sum_{j=2}^N w_j                                           & =1,\tag{adding-up} \label{Ch1_res_add}                                                                                                         \\
	w_j                            \ge 0 \enspace              & \text{for} \enspace j=2,\dots,N,\tag{non-negativity} \label{Ch1_res_non-negativity}                                                            \\
	\sum_{j=2}^Nw_j\boldsymbol{X}_j=\boldsymbol{X}_1  \enspace & \text{and} \enspace \sum_{j=2}^Nw_jY_{jt}=Y_{1t} \enspace \text{for all} \enspace t\le T_0. \tag{pretreatment-fit} \label{Ch1_res_perfect-fit}
\end{align*}

\begin{assume}[]\label{Ch1_assume_fit}
	There exists a set of weights $\left(w_2^*,\dots,w_N^*\right)$ that satisfy the \ref{Ch1_res_add}, \ref{Ch1_res_non-negativity} and \ref{Ch1_res_perfect-fit} restrictions.
\end{assume}

\begin{remark}
	\normalfont
	$\left(w_2^*,\dots,w_N^*\right)$ are random quantities that depend on the sample.
	Assumption \ref{Ch1_assume_fit} is satisfied if the observed predictors and the pretreatment outcomes of the treated unit fall inside the convex hull of those for the control units, in which case there is either a unique or infinitely many sets of weights that satisfy the restrictions, because if there exist two different sets of weights that satisfy the restrictions, then any convex combination of them also satisfy the restrictions.\footnote{The treated unit usually does not fall inside the convex hull of the control units in terms of the matching variables, unless the number of control units is much larger than the number of matching variables, as discussed in Section \ref{Ch1_sec_discuss}.}
	To ensure there is a single solution in this case, \cite{abadie2020penalized} propose a penalised synthetic control method, which adds penalty terms weighted by the pairwise matching discrepancies between the treated unit and the control units, to the problem of minimising the matching discrepancy between the treated unit and the synthetic control.
	Note that the solution of the penalised optimisation problem may not belong to the solution set of the original optimisation problem. For a method that pick the solution that minimises the pairwise matching discrepancies from the sets of weights that satisfy Assumption \ref{Ch1_assume_fit}, see the bilevel optimisation estimator in \cite{diaz2015matching}.
\end{remark}

\medskip

The synthetic control estimator for $\tau_{1t}$ is constructed as
\begin{align}\label{Ch1_eq_SCM}
	\hat{\tau}_{1t}=Y_{1t}-\sum_{j=2}^N w_j^*Y_{jt}.
\end{align}

Before proceeding to the main theoretical result of \cite{abadie2010synthetic}, we also need the following assumption, which ensures that matching on the observed predictors and pretreatment outcomes implies matching on the unobserved predictors.

\begin{assume}[]\label{Ch1_assume_rank}
	The smallest eigenvalue of $\frac{1}{T_0}\sum_{t=1}^{T_0}\boldsymbol{\lambda}_t\boldsymbol{\lambda}_t'$ is bounded from below by some positive number $\underline{\xi}$.
\end{assume}

\medskip

The following theorem gives the main theoretical result in \cite{abadie2010synthetic}, which shows that the bias of the synthetic control estimator goes to zero as the number of pretreatment periods goes to infinity, under the stated assumptions.

\begin{theorem}\label{Ch1_thm_SC}
	Under Assumptions \ref{Ch1_assume_IFE}, \ref{Ch1_assume_error}, \ref{Ch1_assume_fit} and \ref{Ch1_assume_rank},
	$\mathbb{E}\left(\hat{\tau}_{1t}-\tau_{1t}\right)\rightarrow 0 \enspace \text{as} \enspace T_0\rightarrow \infty.$
\end{theorem}

\begin{remark}
	\normalfont
	If we construct the synthetic control by matching only on the observed predictors, then the estimation suffers from the omitted variable bias since the unobserved predictors are not included. Theorem \ref{Ch1_thm_SC} shows that by matching on the observed predictors and the pretreatment outcomes, the unobserved predictors are implicitly matched as well under the stated assumptions.
	The intuition is that if the treated unit and the synthetic control have very different underlying predictors, then it is unlikely that they would match well on all the pretreatment outcomes as $T_0\rightarrow \infty$ simply due to the random noises.
\end{remark}

\begin{remark}
	\normalfont
	The synthetic control estimator constructed by matching only on the pretreatment outcomes can also be shown to be asymptotically unbiased, albeit with a larger bound on the bias, as shown by \cite{botosaru2019role}.
\end{remark}

\medskip

In practice, there may not be a set of weights that satisfy the restrictions in Assumption \ref{Ch1_assume_fit} exactly, and the weights are chosen as
\begin{align*}
	\left(\widetilde{w}_2,\dots,\widetilde{w}_N\right)=\arg & \min_{\left(w_2,\dots,w_N\right)}\sum_{m=1}^Kv_m\left(Z_{1m}-\sum_{j=2}^NZ_{jm}w_j\right)^2 \numberthis\label{Ch1_eq_op} \\
	                                    & s.t. \enspace \sum_jw_j=1 \enspace\text{and}\enspace w_j\ge0,
\end{align*}
where $Z_{i1},\dots,Z_{iK}$ are the $K$ pretreatment variables to match on, and $v_1,\dots,v_K$ are the weights assigned to these variables, representing the importance of each variable in determining the outcomes.
The pretreatment matching variables may include the observed predictors and the pretreatment outcomes, or some linear combinations of them, e.g., the mean of the observed predictors or the pretreatment outcomes across some pretreatment periods.
$v_1,\dots,v_K$ can be chosen by minimising the mean squared prediction errors in the pretreatment periods, with the option of using cross-validation by splitting the pretreatment periods into a training set and a validation set (see \citealp{abadie2015comparative} and \citealp{abadie2021jel} for details).

Inference for the synthetic control method is based on the permutation test, where the treatment is recursively reassigned to each of the control units, and a synthetic control is constructed to predict the outcomes for the control unit using all the other units including the treated unit. The ratio between the posttreatment RMSPE (root mean squared prediction error) $\left[\frac{1}{T-T_0}\sum_{t=T_0+1}^{T}\left(Y_{jt}-\hat{Y}_{jt}\right)\right]^{1/2}$ and the pretreatment RMSPE $\left[\frac{1}{T_0}\sum_{t=1}^{T_0}\left(Y_{jt}-\hat{Y}_{jt}\right)\right]^{1/2}$ is obtained for each unit, and the distribution of the post/pretreatment RMSPE ratios is used for inference, where a large ratio for the treated unit relative to the control units is considered evidence that the treatment effect is statistically significant.

\subsection{Discussion}\label{Ch1_sec_discuss}

The non-negativity restriction is imposed by \cite{abadie2010synthetic} to safeguard against extrapolation, which happens if the values of the predictors for the treated unit fall outside of the convex hull of those of the control units.\footnote{Interpolation happens when the values of the predictors for the treated unit fall in the convex hull.}
However, being in the convex hull does not necessarily translate to nonnegative weights for all control units. According to the Carathéodory's theorem, it is possible for the treated unit in the convex hull of the control units to be represented by a linear combination of the control units, where some control units are assigned negative weights.\footnote{The Carathéodory's theorem states that if a point $Z\in\mathbb{R}^K$ lies in the convex hull of a set of points $P$, where $|P|>K+1$, then $Z$ is in the convex hull of some $K+1$ points in $P$. In other words, $Z$ can be expressed as an affine combination of the points in $P$, where some $K+1$ points are assigned positive weights, while the other points can receive negative weights.}
Furthermore, there is no extrapolation bias when the outcome is a linear function of the underlying predictors, whereas in the presence of nonlinearity, it is more important to use control units that are closer to the treated unit to reduce interpolation bias rather than restricting the weights to be non-negative. To illustrate this, we provide two simple examples in Figure \ref{Ch1_fig_ex}.

\begin{figure}[!htb]
	\centering
	\begin{subfigure}{.48\textwidth}
		\begin{tikzpicture}[scale=0.55]
			\draw[->] (0,0) -- (10,0) node[anchor=north west] {X};
			\draw[->] (0,0) -- (0,6) node[anchor=south east] {$Y^0$};
			\foreach \x in {0,1,2,3,4,5,6,7,8,9}
			\draw[shift={(\x,0)}] (0pt,3pt) -- (0pt,0pt) node[below] {$\x$};
			\foreach \y in {1,2,3,4,5}
			\draw[shift={(0,\y)}] (3pt,0pt) -- (0pt,0pt) node[left] {$\y$};
			\node at (5,.5) (nodeA) {\scriptsize $X_A$};
			\node at (6,.5) (nodeB) {\scriptsize $X_B$};
			\node at (7,.5) (nodeC) {\scriptsize $X_C$};
			\node at (2,.5) (nodeD) {\scriptsize $X_D$};
			\fill[green] (5,0) circle (.8 mm);
			\foreach \i in {2, 6, 7}
			\fill[black] (\i,0) circle (.8 mm);
			\draw[domain=0:9,smooth,variable=\x,black] plot ({\x,{.6*\x}});
			\draw[red,thick,dashed] (5,3) -- (6,3.6) -- (7,4.2);
			\fill[red] (5,3) circle (.6 mm);
			\fill[red] (6,3.6) circle (.6 mm);
			\fill[red] (7,4.2) circle (.6 mm);
			\draw[blue,thick,dashed](2,1.2) -- (5,3) -- (6,3.6);
			\fill[blue] (2,1.2) circle (.6 mm);
			\fill[blue] (5,3) circle (.6 mm);
			\fill[blue] (6,3.6) circle (.6 mm);
			\fill[black] (5,3) circle (.6 mm);
		\end{tikzpicture}
		\caption{Linear outcome} \label{Ch1_fig_ex1}
	\end{subfigure}
	~
	\begin{subfigure}{.48\textwidth}
		\begin{tikzpicture}[scale=0.55]
			\draw[->] (0,0) -- (10,0) node[anchor=north west] {X};
			\draw[->] (0,0) -- (0,6) node[anchor=south east] {$Y^0$};
			\foreach \x in {0,1,2,3,4,5,6,7,8,9}
			\draw[shift={(\x,0)}] (0pt,3pt) -- (0pt,0pt) node[below] {$\x$};
			\foreach \y in {1,2,3,4,5}
			\draw[shift={(0,\y)}] (3pt,0pt) -- (0pt,0pt) node[left] {$\y$};
			\node at (5,.5) (nodeA) {\scriptsize $X_A$};
			\node at (6,.5) (nodeB) {\scriptsize $X_B$};
			\node at (7,.5) (nodeC) {\scriptsize $X_C$};
			\node at (2,.5) (nodeD) {\scriptsize $X_D$};
			\fill[green] (5,0) circle (.8 mm);
			\foreach \i in {2, 6, 7}
			\fill[black] (\i,0) circle (.8 mm);
			\draw[domain=0:10,smooth,variable=\x,black] plot ({\x,{5/(1 + (exp(1*(3-\x))))}});
			\draw[red,thick,dashed] (5,4.615673) -- (6,4.762871) -- (7,4.910069);
			\fill[red] (5,4.615673) circle (.6 mm);
			\fill[red] (6,4.762871) circle (.6 mm);
			\fill[red] (7,4.910069) circle (.6 mm);
			\draw[blue,thick,dashed](2,1.344707) -- (5,3.90833) -- (6,4.762871);
			\fill[blue] (2,1.344707) circle (.6 mm);
			\fill[blue] (5,3.90833) circle (.6 mm);
			\fill[blue] (6,4.762871) circle (.6 mm);
			\draw[red,thick] (5,4.403985) -- (5,4.615673);
			\draw[blue,thick] (5,4.403985) -- (5,3.90833);
			\fill[black] (5,4.403985) circle (.6 mm);
			\node at (3,5) (nodeA) {\textcolor{red}{Bias 1}};
			\node at (5,4.5) (nodeB) {};
			\node at (3,4) (nodeC) {\textcolor{blue}{Bias 2}};
			\node at (5,4.2) (nodeD) {};
			\draw[->] (nodeA) -- (nodeB) node [left] {};
			\draw[->] (nodeC) -- (nodeD) node [left] {};
		\end{tikzpicture}
		\caption{Nonlinear outcome} \label{Ch1_fig_ex2}
	\end{subfigure}
	\caption{Example} \label{Ch1_fig_ex}
\end{figure}

In both examples, we assume a single observed predictor $X$ of the untreated potential outcome $Y^0$. There is one treated unit $A$ with $X_A=5$ whose untreated potential outcome is not observed, and which we wish to estimate using the outcomes of the control units. Suppose that there are only two control units $B$ and $C$ with $X_B=6$ and $X_C=7$, then we are not able to construct a synthetic control that perfectly matches the treated unit with the non-negativity restriction imposed, even though $X_A=5$ is just outside the convex hull of $X_B=6$ and $X_C=7$. If a third control unit $D$ with $X_D=1$ is available, we can construct a synthetic control using $X_D=1$ and $X_B=6$ under the non-negativity restriction, and compare it with the synthetic control constructed using $X_B=6$ and $X_C=7$ without the non-negativity restriction, both of which perfectly match the treated unit.
In Figure \ref{Ch1_fig_ex1}, the untreated potential outcome is a linear function of the predictor given by $Y^0=0.6X$. We see that there is no extrapolation bias with or without the non-negativity restriction, as both synthetic controls provide perfect estimates for the counterfactual outcome of the treated unit $Y_A^0=3$.
In Figure \ref{Ch1_fig_ex2}, the untreated potential outcome is a nonlinear function of the predictor given by the S-shaped logistic function $Y^0=\frac{5}{1+e^{3-X}}$.\footnote{The intuition applies to other nonlinear functions that satisfy Assumption \ref{Ch1_assume_nY}.} We see that the magnitude of the extrapolation bias ($|\text{Bias}\ 1|\approx 0.21$) for the synthetic control estimator constructed without the non-negativity restriction is smaller than that of the interpolation bias ($|\text{Bias}\ 2|\approx 0.50$) for the synthetic control estimator constructed with the non-negativity restriction, since the former uses closer neighbours.
The moral of this example is that in the presence of nonlinearity, the interpolation bias or the extrapolation bias tends to be smaller if we construct the synthetic control using control units that are closer to the treated unit, whereas the original synthetic control method with the non-negativity restriction does not prefer or implement the use of closer neighbours. For example, suppose that there is an additional control unit $E$ with $X_E=3$, the original synthetic control method does not have a preference on $\frac{1}{3}E+\frac{2}{3}B$ over $\frac{1}{4}D+\frac{3}{4}B$ since both match $X_A$ perfectly with positive weights, even though the bias of the former is smaller in the nonlinear case.

The non-negativity restriction also makes it less likely to obtain a set of weights that satisfy the \ref{Ch1_res_add}, \ref{Ch1_res_non-negativity} and \ref{Ch1_res_perfect-fit} restrictions (Assumption \ref{Ch1_assume_fit}).
Note that without the non-negativity restriction, we need $J\ge L=1+k+T_0$ to be able to find a set of weights that match the synthetic control and the treated unit perfectly, since otherwise the matching variables of the control units do not span $\mathbb{R}^{L}$. For the weights to satisfy the additional \ref{Ch1_res_non-negativity} restriction, it is likely that $J$ needs to be much larger than $L$, especially when $L$ is large due to the curse of dimensionality.
As an example, we generate 1000 samples, with a single treated unit, 10,000 control units and 10 pretreatment periods in each sample. A typical sample is shown in Figure \ref{Ch1_fig_hull1}, where the trajectory of the outcome for the treated unit is depicted in black and the trajectories for 20 of the control units are in gray.\footnote{For the detailed data generating process, see Section \ref{Ch1_sec_sim}.} The levels of the outcome for the units are relatively stable over time, which is similar to what we observe in real data. To make the treated unit even more likely to be in the convex hull of the control units, we calibrate the level of the outcome for the treated unit so that it is at the mean of the outcomes for the control units in period 1.
We then use the method provided by \cite{king2006dangers}, where the convex hull membership check problem is characterised as a linear programming problem, to check whether the treated unit is in the convex hull of the control units.
Figure \ref{Ch1_fig_hull2} records the median sample size required for the treated unit to be in the convex hull of the control units in terms of the corresponding number of pretreatment outcomes. The result suggests that with the non-negativity restriction, the sample size needed for a perfect fit on only a few number of pretreatment periods already exceeds the sample sizes usually available for the synthetic control method.

\begin{figure}[!htb]
	\centering
	\begin{subfigure}{.48\textwidth}
		\centering
		\includegraphics[width=\linewidth]{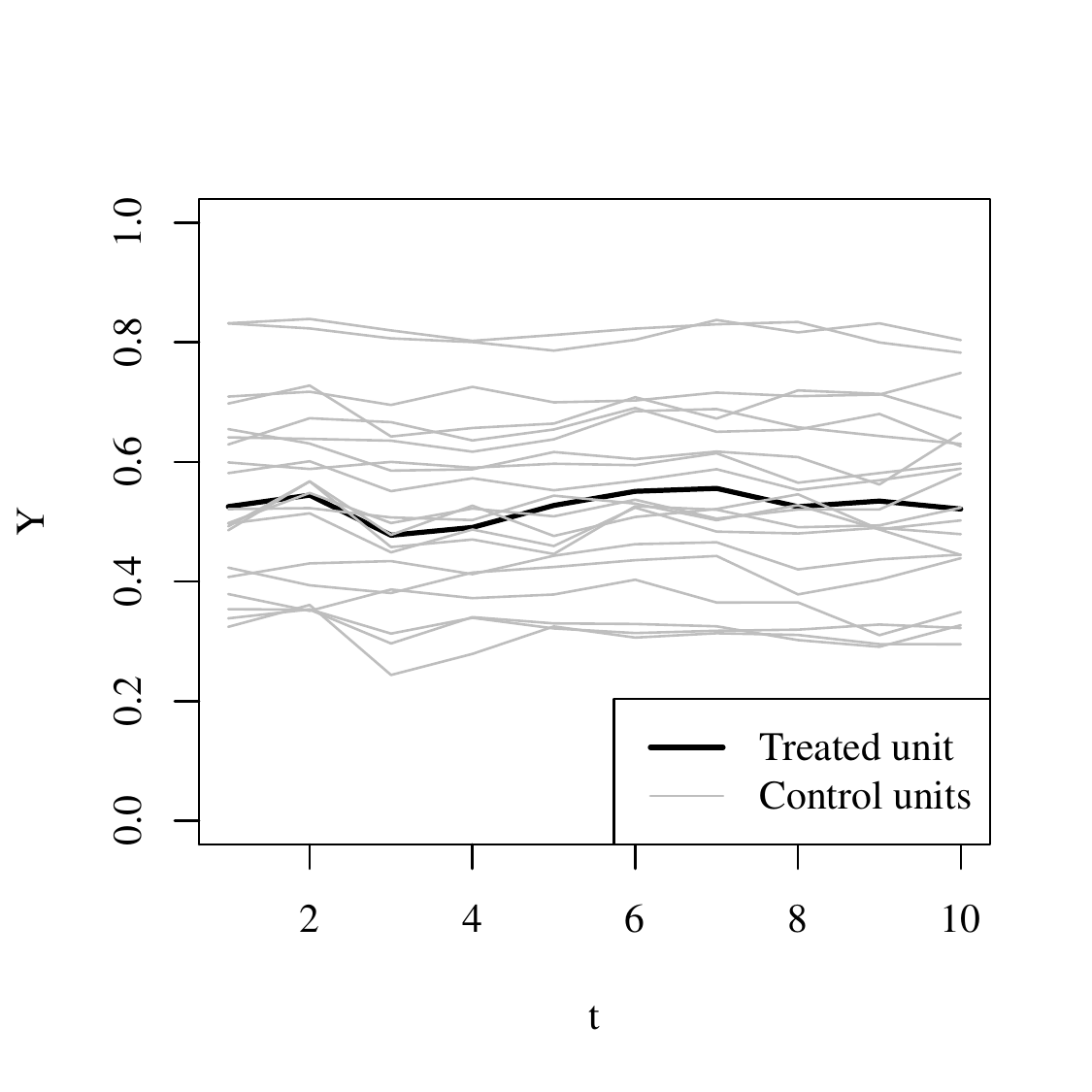}
		\caption{Typical sample}
		\label{Ch1_fig_hull1}
	\end{subfigure}
	~
	\begin{subfigure}{.48\textwidth}
		\centering
		\includegraphics[width=\linewidth]{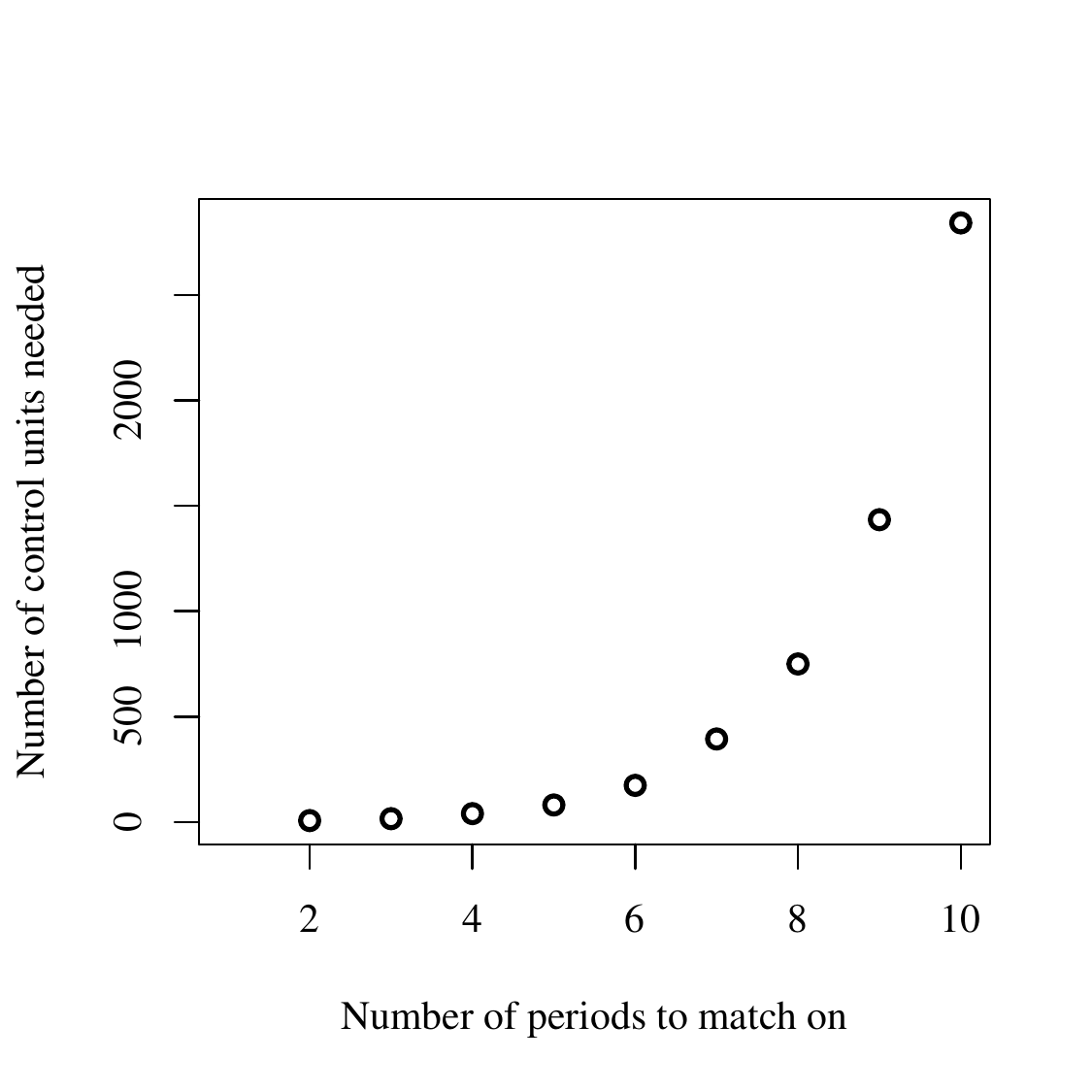}
		\caption{Sample size needed}
		\label{Ch1_fig_hull2}
	\end{subfigure}
	\caption{Simulated Example}
	\label{Ch1_fig_hull}
\end{figure}

Even if we were to construct a synthetic control that matches the treated unit only approximately, the non-negativity restriction makes a good pretreatment fit less likely, especially when the treated unit takes extreme values in the matching variables or when the sample size is small, which limits the applicability of the synthetic control method.\footnote{In cases where the treated unit can be closely approximated by the synthetic control, the non-negativity restriction often ensures that only a few control units receive positive weights, which is an appealing feature in comparative case studies, since it makes it easier to interpret the contribution of each control unit in the construction of the synthetic control.}
\cite{abadie2010synthetic,abadie2015comparative} recommend using the synthetic control method only when the treated unit can be closely approximated by the synthetic control. However, since the units near the boundary of the distribution may not be well approximated by the synthetic controls constructed using the other units with the non-negativity restriction, this indicates that the post/pretreatment RMSPE ratio in the permutation test is conditional on a good pretreatment fit for the treated unit but unconditional for the others, which would lead to over-rejection of the null hypothesis, as noted in \cite{ferman2017placebo}.

In light of the discussion, we relax the non-negativity restriction so that a good pretreatment fit is more likely to be obtained for the treated unit as well as for the other units when conducting inference. This not only expands the applicability of the synthetic control method, but also helps correct the size distortion for inference. In addition, with the non-negativity restriction lifted, we may be able to obtain a synthetic control estimator with a smaller bias using more flexible regularisation methods than if the restriction were imposed, since the solution space of the weights in the former is a superset of that in the latter.

\section{The Synthetic Control Method with Nonlinearity}\label{Ch1_sec_NSC}

The synthetic control estimator is shown to be asymptotically unbiased in \cite{abadie2010synthetic}, provided that the outcome is a linear function of the underlying predictors and that the pretreatment matching variables of the treated unit can be well approximated by those of the synthetic control. When the outcome is nonlinear, however, the bias of the synthetic control estimator may be severe, since a good fit on the pretreatment matching variables between the treated unit and the synthetic control does not necessarily imply a good fit on the unobserved predictors.
In this section, we provide the conditions for the synthetic control estimator to be asymptotically unbiased when the outcome is nonlinear, and propose a flexible and data-driven method for choosing the synthetic control weights in practice.


We start by assuming the following conditions, which are adapted from the assumptions for the matching estimator in \cite{abadie2006large}.

\begin{assume}[]\label{Ch1_assume_dist}
	\leavevmode
	\begin{enumerate}[label=\arabic*)]
		\item Let $\boldsymbol{H}=[\boldsymbol{X}' \enspace \boldsymbol{\mu}']'$ be a $\left(k+f\right)\times 1$ random vector of continuous variables, with a version of the density $c<f\left(\boldsymbol{H}\right)<d$ for some $c,d>0$ on its compact and convex support $\mathbb{H}\in\mathbb{R}^{k+f}$;
		\item $\{\boldsymbol{H}_i\}_{i=1}^N$ are independent draws from the distribution of $\boldsymbol{H}$;
		\item For almost every $\boldsymbol{h}\in\mathbb{H}$, $\text{Pr}\left(D_{i,T_0+1}=1\mid\boldsymbol{H}_i=\boldsymbol{h}\right)<1-\rho$ for some $\rho\in\left(0,1\right)$.
	\end{enumerate}
\end{assume}

\begin{remark}
	\normalfont
	Assumption \ref{Ch1_assume_dist} ensures that the observed and unobserved predictors for all units can be drawn independently from almost any point in the support, and that for almost any values that the predictors of a treated unit take, it is possible to have a control unit whose predictors take those values.
	This assumption is stronger than Assumption \ref{Ch1_assume_fit}, which restricts the application of the synthetic control method to samples where a synthetic untreated ``twin'' for the treated unit can be constructed as a linear combination of the control units. Assumption \ref{Ch1_assume_dist} assumes the existence of individual untreated near-identical twins in the population, and that if the random sample is large enough some of those twins will be in the sample. This stronger assumption is needed in the presence of nonlinearity.
\end{remark}

\medskip

The untreated potential outcome is assumed to be linked with the linear latent outcome through a strictly monotonic function, and that its expectation with respect the individual transitory shock is a smooth function.

\begin{assume}\label{Ch1_assume_nY}
	$Y_{it}^0=F\left(\boldsymbol{X}_i'\boldsymbol{\beta}_t+\boldsymbol{\mu}_i'\boldsymbol{\lambda}_t+\varepsilon_{it}\right)$, where $F\left(\cdot\right)$ is a strictly monotonic function.
	$\mathbb{E}_{\varepsilon}\left(Y_{it}^0\right)=G\left(\boldsymbol{X}_i'\boldsymbol{\beta}_t+\boldsymbol{\mu}_i'\boldsymbol{\lambda}_t\right)$, where $\mathbb{E}_{\varepsilon}\left(\cdot\right)$ is the expectation conditional on $\boldsymbol{X}_i$ and $\boldsymbol{\mu}_i$, and $G\left(\cdot\right)$ is a smooth function.
\end{assume}

\begin{remark}
	\normalfont
	Assuming that $F\left(\cdot\right)$ is a strictly monotonic function excludes binary outcomes since they are usually modelled by discrete choice models like probit or logit, where different values of the latent outcomes can lead to the same observed outcome. Matching on these pretreatment outcomes only implies that the latent outcomes are in the same interval, and there is no guarantee that the unobserved predictors are matched.
	The unknown functions $F\left(\cdot\right)$ and $G\left(\cdot\right)$ need not to be estimated as long as the regularity conditions in Assumption \ref{Ch1_assume_nY} are satisfied, since our goal is to provide conditions for the synthetic control estimator to be asymptotically unbiased when the outcome has an unknown general nonlinear functional form.
\end{remark}

\medskip

Let $\mathbbm{1}\left(\cdot\right)$ be the indicator function, $\Vert\cdot\Vert$ be the Euclidean norm, $\boldsymbol{Z}_i=[\boldsymbol{X}_i' \enspace Y_{i1} \enspace \cdots \enspace Y_{iT_0}]'$, and $\mathcal{J}=\{2,\dots,N\}$, then the set of indices for the $M$ closest neighbours of the treated unit in terms of the observed predictors and the pretreatment outcomes can be denoted as
\begin{equation}\label{Ch1_eq_JM}
	\mathcal{J}_M=\left\{j\in \mathcal{J} \middle| \enspace \sum_{l\in \mathcal{J}} \mathbbm{1}\left(\Vert \boldsymbol{Z}_l-\boldsymbol{Z}_1 \Vert < \Vert \boldsymbol{Z}_j-\boldsymbol{Z}_1 \Vert\right) < M\right\},
\end{equation}
i.e., $\Vert \boldsymbol{Z}_l-\boldsymbol{Z}_1 \Vert > \Vert \boldsymbol{Z}_j-\boldsymbol{Z}_1 \Vert$ for any $l\in\mathcal{J}\backslash\mathcal{J}_M$ and $j\in\mathcal{J}_M$, where $\backslash$ takes the difference of two sets.

\begin{assume}\label{Ch1_assume_neighbor}
	Only the nearest $M>k+T_0$ neighbours are used for constructing the synthetic control, i.e.,
	$\sum_{j\in\mathcal{J}_M}w_j^*=1$, $\sum_{j\in\mathcal{J}_M} w_j^*\boldsymbol{X}_j=\boldsymbol{X}_1$ and $\sum_{j\in\mathcal{J}_M} w_j^*Y_{jt}=Y_{1t}$ for all $t\le T_0$, and $w_j^*=0$ for all $j\not\in\mathcal{J}_M.$
\end{assume}

\medskip

Using the nearest $M$ neighbours, we can construct the synthetic control estimator as $\tilde{\tau}_{1t}=Y_{1t}-\sum_{j\in\mathcal{J}_M} w_j^*Y_{jt}$. The following theorem provides the conditions for the synthetic control estimator to be asymptotically unbiased when the outcome is nonlinear.

\begin{theorem}\label{Ch1_thm_NSC}
	Under Assumptions \ref{Ch1_assume_error}, \ref{Ch1_assume_rank}, \ref{Ch1_assume_dist}, \ref{Ch1_assume_nY} and \ref{Ch1_assume_neighbor},
	$\mathbb{E}\left(\tilde{\tau}_{1t}-\tau_{1t}\right)\rightarrow 0 \enspace \text{as} \enspace T_0\rightarrow \infty$ if $J=O\left(T_0^{b(T_0)}\right)$ with $b(\cdot)\ge 1$ and $b'(\cdot)>0$.
\end{theorem}

\begin{remark}
	\normalfont
	Theorem \ref{Ch1_thm_NSC} states that the synthetic control estimator is asymptotically unbiased if the outcome is a strictly monotonic function with smooth expectation, the synthetic control is constructed using the nearest $M$ neighbours, and that the number of control units increases more than exponentially faster than the number of pretreatment periods.
	The idea is that the bias of the synthetic control estimator will not vanish asymptotically when the outcome is an unknown nonlinear function, unless each control unit used for constructing the synthetic control converges to the treated unit in the underlying predictors, which happens if the synthetic control is constructed using the nearest neighbours, and that the number of control units increases much faster than the number of pretreatment periods so that the nearest neighbours become closer and closer to the treated unit.
	This result does not conflict with the fact that the sample sizes available for the synthetic control method are usually small, as it just provides the conditions for the existence of control units that are similar to the treated unit.
	The implication of this result in finite samples is that when the outcome is highly nonlinear, the bias of the synthetic control estimator is expected to be small if there exist control units that are similar to the treated unit and that we construct the synthetic control using only those control units.\footnote{Note that this is different from the sparsity of the weights under the non-negativity restriction, which is achieved regardless of the degree of nonlinearity, as the few control units that receive positive weights are not necessarily close to the treated unit.}
\end{remark}

\medskip

While Theorem \ref{Ch1_thm_NSC} provides conditions for the synthetic control estimator to be asymptotically unbiased when the outcome is nonlinear, it is also important to examine the bias in finite samples.
From the proof of Theorem \ref{Ch1_thm_NSC}, we have
\begin{align}
	\mathbb{E}_{\varepsilon}\left(\sum_j w_j^*Y_{jt}-Y_{1t}^0\right)
	= \sum_{n=1}^{\infty}\frac{G^{\left(n\right)}}{n!}\left\{\sum_j w_j^*\left[\left(\boldsymbol{X}_j-\boldsymbol{X}_1\right)'\boldsymbol{\beta}_{t}+\left(\boldsymbol{\mu}_j-\boldsymbol{\mu}_1\right)'\boldsymbol{\lambda}_{t}\right]^n\right\},
\end{align}
which represents the weighted average of the distances between the treated unit and the control units in terms of the underlying predictors.
Under the assumptions of Theorem \ref{Ch1_thm_NSC}, $\boldsymbol{Z}_j-\boldsymbol{Z}_1\overset{p}{\rightarrow}0$ implies that $\boldsymbol{\mu}_j-\boldsymbol{\mu}_1\overset{p}{\rightarrow}0$ for $j\in\mathcal{J}_M$.
Therefore, when the outcome is highly nonlinear, the bias of the synthetic control estimator would be smaller if we use fewer and closer neighbours to construct the synthetic control so that the pairwise matching discrepancies in the pretreatment variables between the treated unit and its neighbours are smaller.

When the outcome is linear or if the degree of nonlinearity is low, the bias of the synthetic control estimator reduces to
\begin{align}
	\mathbb{E}_{\varepsilon}\left(\sum_j w_j^*Y_{jt}-Y_{1t}^0\right)
	= \left(\sum_j w_j^*\boldsymbol{X}_j-\boldsymbol{X}_1\right)'\boldsymbol{\beta}_{t}+\left(\sum_j w_j^*\boldsymbol{\mu}_j-\boldsymbol{\mu}_1\right)'\boldsymbol{\lambda}_{t},
\end{align}
representing the distance between the treated unit and the synthetic control in the underlying predictors, and is smaller if we construct the synthetic control using more neighbours so that the matching discrepancy between the treated unit and the synthetic control is smaller. The bias is also shown in the proof of Theorem \ref{Ch1_thm_SC} to be bounded by a value that increases with $\bar{w}=\text{max}_j|w_j^*|$ given a good pretreatment fit between the treated unit and the synthetic control constructed by $J$ control units, and thus is smaller if the weights are assigned more evenly among the control units. The variance of the synthetic control estimator also tends to be smaller if the weights are more spread out, which is similar to the least square estimator, whose variance becomes smaller if the sample size is larger or if the explanatory variables are more spread out.

This presents a trade-off between the aggregate matching discrepancy and the pairwise matching discrepancies, depending on the degree of nonlinearity of the outcome function, similar to the bias-variance tradeoff in non-parametric methods, e.g., choosing bin width in kernel density estimation.
To address this trade-off, we choose the set of weights by solving the following minimisation problem with elastic net type penalties,
\begin{align*}
	\min_{\{w_j\}_j} & \left\Vert \boldsymbol{Z}_1-\sum_j w_j\boldsymbol{Z}_j \right\Vert^2+a\sum_j \left\vert w_j \right\vert \left\Vert \boldsymbol{Z}_1-\boldsymbol{Z}_j \right\Vert+b\sum_j \left\vert w_j \right\vert^2, \numberthis \\
	                 & \text{s.t.} \sum_j w_j=1.
\end{align*}

The $L_1$ penalty terms are weighted by pairwise matching discrepancies between the treated unit and the control units, and penalise assigning weights to control units that are farther away from the treated unit. The $L_2$ penalty term penalises concentrating weights on a few control units and controls the scale of $\bar{w}$.
The level of penalisation is adjusted through the nonnegative tuning parameters, $a$ and $b$.
When $a=0$ and $b=0$, the weights are chosen solely to minimise the aggregate matching discrepancy.
When $a$ becomes larger, the weights are more concentrated on control units that are closer to the treated unit, thus achieving sparsity of the weights. As $a\rightarrow\infty$, the estimator becomes the nearest neighbours matching estimator using only the nearest neighbour, as noted in \cite{abadie2020penalized}.
When $b$ becomes larger, the weights are assigned more evenly among the control units. As $b\rightarrow\infty$, all the control units are assigned equal weights and the estimator becomes the difference-in-differences estimator.
When both $a$ and $b$ are large, the weights are spread out among a number of control units that are close to the treated unit, and the estimator becomes close to the nearest neighbours matching estimator using multiple neighbours.

Ultimately choosing the optimal tuning parameters is an empirical problem in finite samples, which can be done using cross-validation, as proposed in \cite{doudchenko2017} and \cite{abadie2020penalized}.
One way to conduct cross-validation is to predict the posttreatment outcomes for each control unit using the synthetic control constructed from the other control units, and the optimal set of tuning parameters is the one that minimises the mean squared prediction error.
This method of conducting cross-validation is used for the Monte Carlo simulations and the applications in this paper.
Alternatively, we can predict the outcome of the treated unit in each pretreatment period using the synthetic control constructed from the control units, and select the set of tuning parameters that minimises the mean squared prediction error.
To make the selection of the optimal tuning parameters tractable in practice, the tuning parameters that enter the minimisation problem, $a$ and $b$, are scaled by the nonzero eigenvalues of $\boldsymbol{Z}_0\boldsymbol{Z}_0'$, where $\boldsymbol{Z}_0$ is the $J\times (k+T_0)$ matrix of matching variables of the control units, so that the optimal tuning parameters $a^*$ and $b^*$ can be chosen from $[0,1]$. Specifically, $\boldsymbol{Z}_0\boldsymbol{Z}_0'$ has $n=\min(J,\ k+T_0)$ nonzero eigenvalues, denoted as $\lambda_1,\dots,\lambda_n$, in ascending order.
For $b^*\in[0,1]$, we set $b=b^*\lambda_{\lceil nb^* \rceil}$, where $\lceil \cdot \rceil$ is the ceiling function, so that when $a^*=0$ and $b^*=1$, the weights are roughly evenly assigned to the control units. $a$ is similarly scaled by the nonzero eigenvalues of $\boldsymbol{Z}_0\boldsymbol{Z}_0'+\text{diag}(b)$, so that when $a^*=1$, the weight will only be assigned to the nearest neighbour.
To select the optimal tuning parameters, we start with the initial value $b^*=0$, and then choose $a^*$ from $[0,1]$ with a set grid size, e.g., 0.1, to minimise the mean squared prediction error from either cross-validation construction. Given the selected $a^*$, we then update $b^*$ by minimising the mean squared prediction error. $a^*$ and $b^*$ are then updated iteratively until convergence.

\section{Monte Carlo Simulations}\label{Ch1_sec_sim}

In this section, we conduct Monte Carlo simulations to compare the  nonlinear synthetic control estimator (NSC) with the original synthetic control estimator (OSC) from \cite{abadie2010synthetic}, the synthetic control estimator with elastic net regularisation (ESC) from \cite{doudchenko2017}, and the penalised synthetic control estimator (PSC) from \cite{abadie2020penalized}. For the purpose of comparison, these other methods are modified so that OSC differs from NSC in that it imposes the non-negativity restriction and does not have $L_1$ and $L_2$ penalties, ESC differs from NSC in that the $L_1$ penalty terms are not weighted by pairwise matching discrepancies, and PSC differs from NSC in that it does not have $L_2$ penalty.
The number of treated unit is fixed at 1, and the number of posttreatment period is fixed at 10 across settings.
The data generating process is as follows.

First, the latent outcomes are generated from the interactive fixed effects model as
\begin{equation}
	Y_{it}^*=\boldsymbol{X}_i'\boldsymbol{\beta}_t+\boldsymbol{\mu}_i'\boldsymbol{\lambda}_t+\varepsilon_{it},
\end{equation}
where the vector of observed predictors $\boldsymbol{X}_i$ has dimension 2, and the vector of unobserved predictors $\boldsymbol{\mu}_i$ has dimension 4.
The observed and unobserved predictors are independently and identically drawn from the uniform distribution $U\left[0,2\sqrt{3}\right]$ for each unit, the coefficients are i.i.d. $N(10,1)$, and the individual transitory shocks are i.i.d. $N(0,1)$.

The untreated potential outcomes are then generated as
\begin{equation}\label{Ch1_eq_tran}
	Y_{it}^0=\left(\frac{Y^*_{it}-Y^*_{\text{min}}}{Y^*_{\text{max}}-Y^*_{\text{min}}}\right)^r,
\end{equation}
where $Y^*_{\text{min}}$ and $Y^*_{\text{max}}$ are the smallest and largest values of $Y^*_{it}$ respectively, so that $\frac{Y^*_{it}-Y^*_{\text{min}}}{Y^*_{\text{max}}-Y^*_{\text{min}}}$ is between 0 and 1.\footnote{The transformation in \eqref{Ch1_eq_tran} is random since $Y^*_{\text{min}}$ and $Y^*_{\text{max}}$ depend on the sample. The purpose of rescaling the outcomes to be within $[0,1]$ is to make results in different settings more comparable, and the findings do not fundamentally change if we use some non-random transformation.} 
The degree of nonlinearity is adjusted by $r\in \{1,2\}$. $Y_{it}^0$ is a linear function of the predictors and the individual transitory shock when $r=1$, and is nonlinear when $r=2$.\footnote{Using a larger $r$ or adopting other functional forms such as the logistic function considered in Figure \ref{Ch1_fig_ex2} does not fundamentally change the conclusion.}

The treatment effects $\tau_{it}$ are set to $[0.02,0.04,\dots,0.2]$ in the 10 posttreatment periods, and 0 in the pretreatment periods.
And the observed outcomes are generated as
\begin{equation}
	Y_{it}=Y_{it}^0+D_{it}\tau_{it},
\end{equation}
where
\begin{equation*}
	D_{it}=
	\begin{cases}
		1, & \text{if}\ i=1\ \text{and}\ t>T_0, \\
		0, & \text{otherwise}.
	\end{cases}
\end{equation*}

By varying the number of control units $J\in \{25,50\}$, the number of pretreatment periods $T_0\in \{15,30\}$ and the degree of nonlinearity $r\in \{1,2\}$, we have 8 settings. The observed and unobserved predictors and their coefficients are drawn 20 times for each setting, and the individual transitory shocks are drawn 250 times for each set of $(\boldsymbol{X}_i,\boldsymbol{\beta}_t,\boldsymbol{\lambda}_t,\boldsymbol{\mu}_i)$, so that we generate 5000 samples for each setting.\footnote{To save computation time, for each set of $(\boldsymbol{X}_i,\boldsymbol{\beta}_t,\boldsymbol{\lambda}_t,\boldsymbol{\mu}_i)$, the tuning parameters are chosen once from $[0,1]$ with grid size 0.1 using cross-validation, and are then fixed across the 250 simulations.}

\begin{figure}[!htb]
	\centering
	\begin{subfigure}{\textwidth}
		\centering
		\includegraphics[width=.49\linewidth]{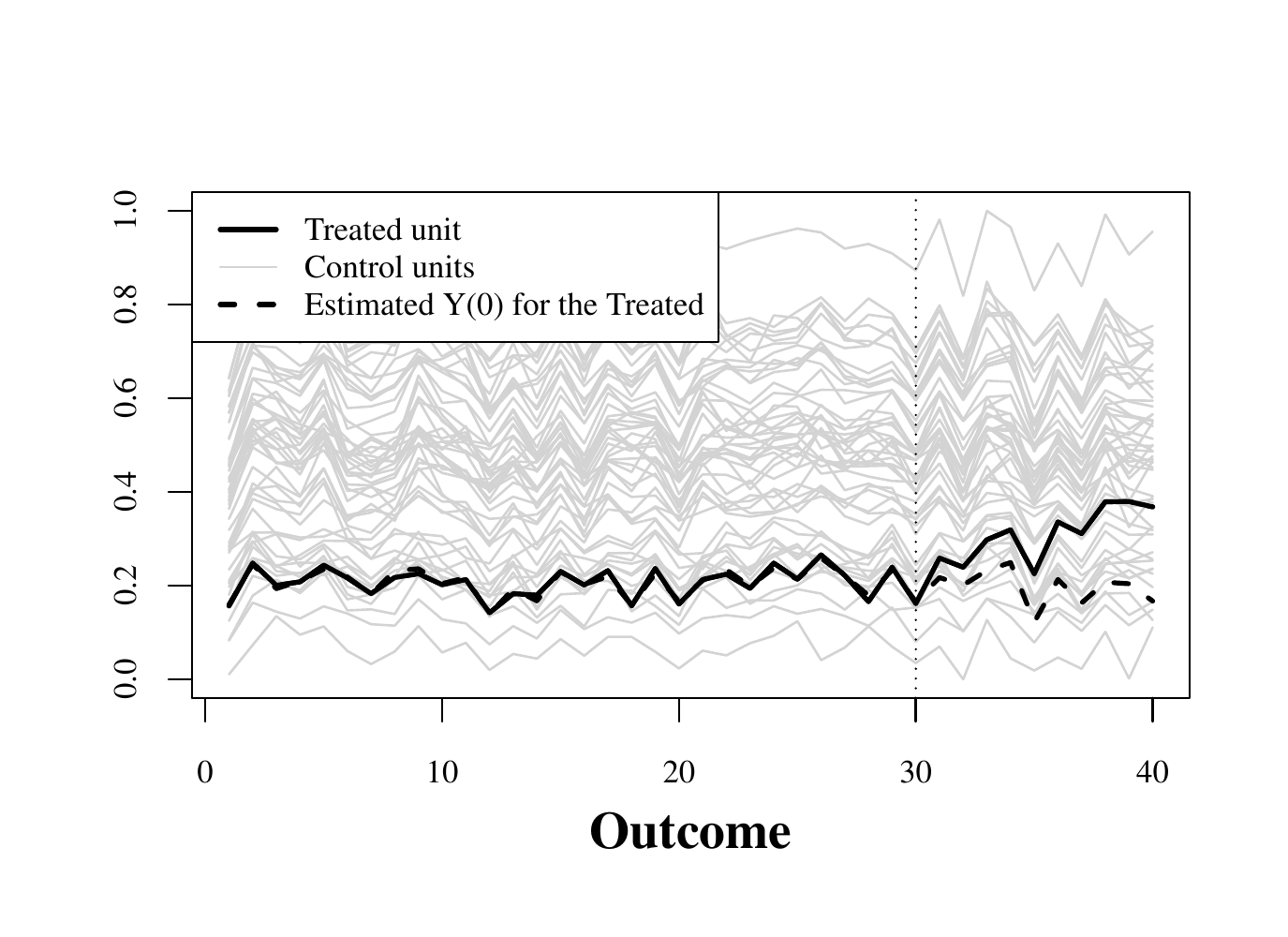}
		\includegraphics[width=.49\linewidth]{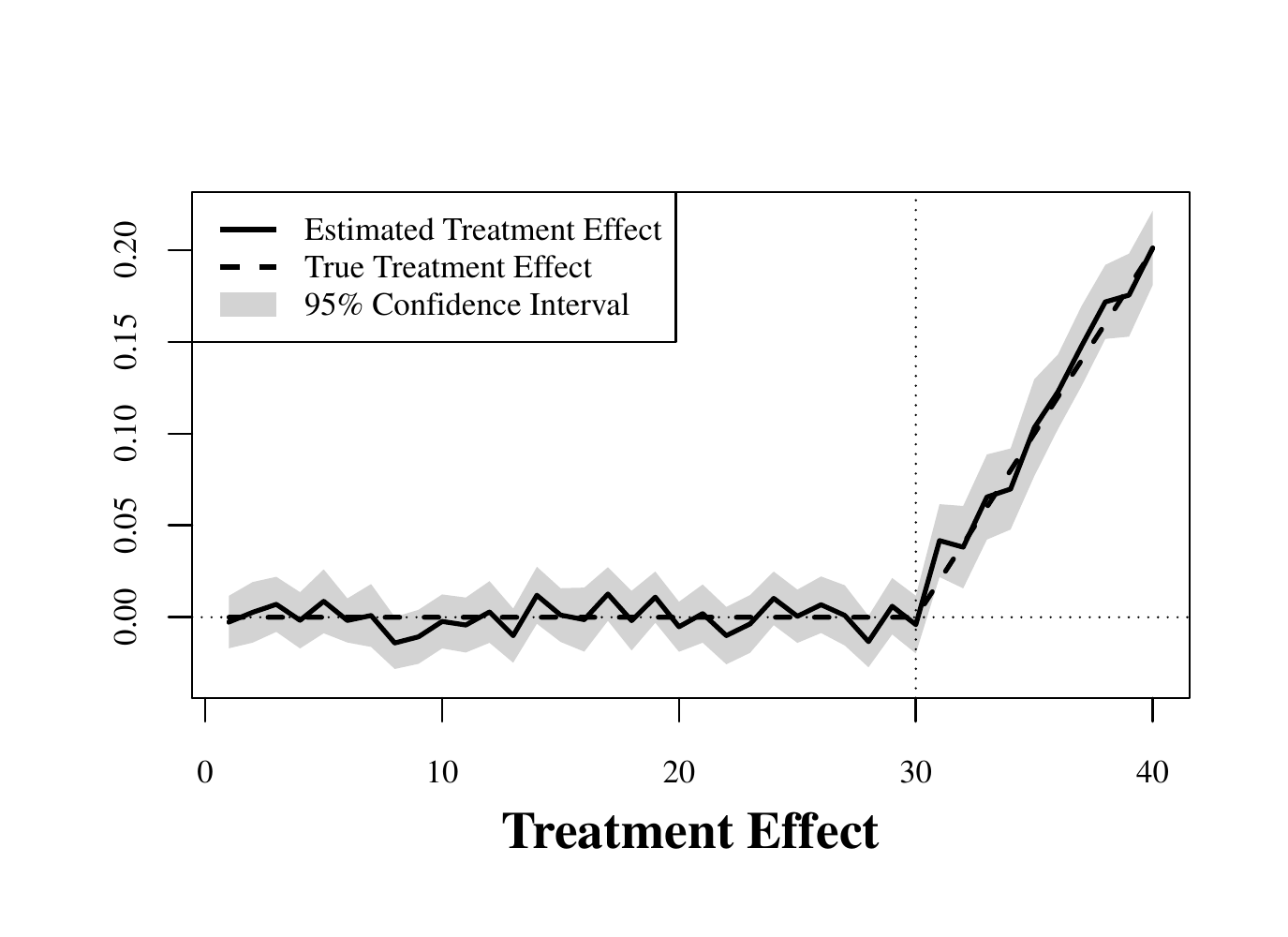}
		\caption{Linear outcome}
		\label{Ch1_fig_sim_ex1}
	\end{subfigure}
	\\
	\begin{subfigure}{\textwidth}
		\centering
		\includegraphics[width=.49\linewidth]{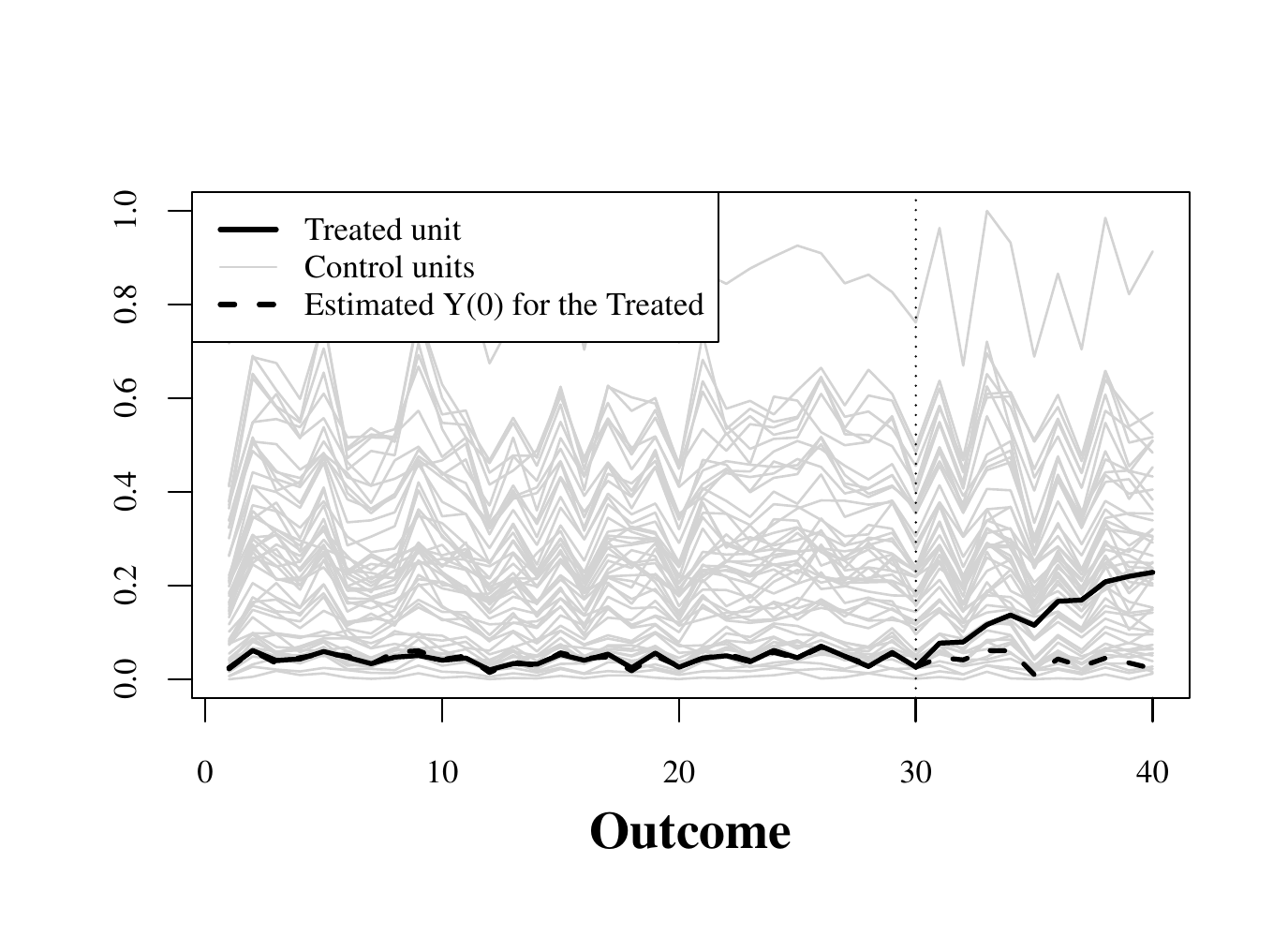}
		\includegraphics[width=.49\linewidth]{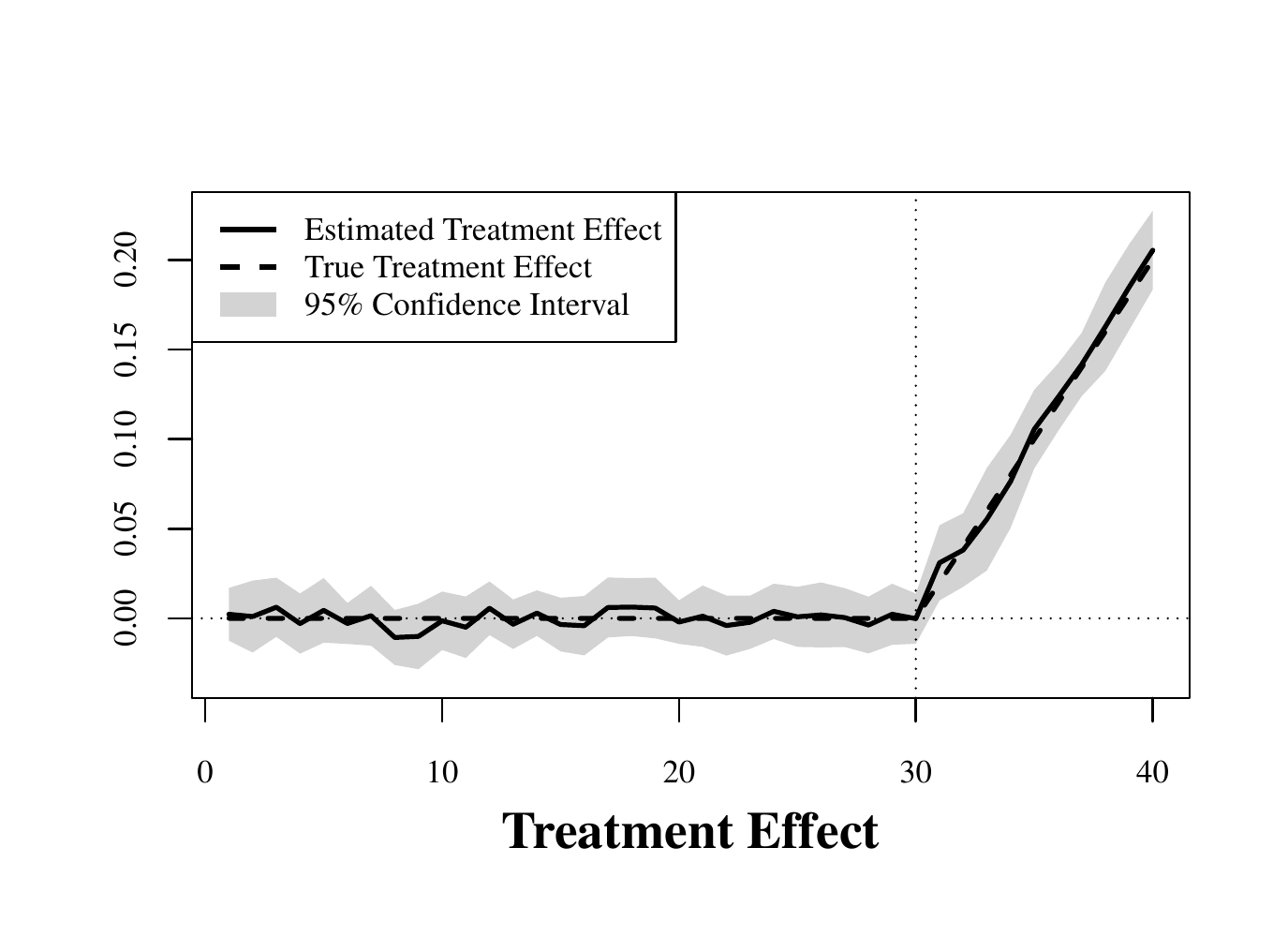}
		\caption{Nonlinear outcome}
		\label{Ch1_fig_sim_ex2}
	\end{subfigure}
	\caption{Simulated Example}
	\label{Ch1_fig_sim_ex}
\end{figure}

Figure \ref{Ch1_fig_sim_ex} illustrates the estimation of the treatment effects using the nonlinear synthetic control method in a typical sample with $J=50$, $T_0=30$ in the linear case ($r=1$, upper panel) and nonlinear case ($r=2$, lower panel), respectively. The graphs on the left visualise the trajectories of the outcome for the treated unit (black solid line), the control units (gray solid lines) and the synthetic control (black dashed line). The observed outcomes are more concentrated towards the bottom in the nonlinear case due to the nonlinear transformation.
In both examples, we see that the trajectory of the outcome for the synthetic control, which is used to estimate the untreated potential outcome for the treated unit, is able to follow the trajectory of the treated unit closely before the treatment, and diverges after the treatment.
The gap between the trajectories of the treated unit and the synthetic control is then used to estimate the treatment effect, as depicted in the graphs on the right, where the black solid line is the estimated treatment effect, the black dashed line is the true treatment effect, and the 95\% confidence interval is in gray.\footnote{To construct the confidence interval, we follow \cite{doudchenko2017} and estimate the variance of the estimator in each period using the mean squared error obtained from predicting the outcome for each control unit in that period using the other control units. The confidence interval is then constructed using the estimated variance, assuming normal distribution for the estimator.} We see that the estimated effects are very close to the true treatment effects, and the confidence intervals also accurately reveal that the treatment effect is not statistically significantly different from 0 before the treatment, and becomes significant after the treatment.

\begin{center}
	\resizebox{13cm}{!}{
		\begin{threeparttable}
			\centering
			\caption{Monte Carlo Results}\label{Ch1_tab_sim}
			\begin{tabular}{ccc@{\hskip 0.5ex}ccc@{\hskip 0.5ex}ccc@{\hskip 0.5ex}ccc@{\hskip 0.5ex}cccc}
				\toprule
				\multicolumn{4}{c}{} & \multicolumn{2}{c}{OSC} & \multicolumn{1}{c}{} & \multicolumn{2}{c}{ESC} & \multicolumn{1}{c}{} & \multicolumn{2}{c}{PSC} & \multicolumn{1}{c}{} & \multicolumn{3}{c}{NSC}                                                     \\
				\cmidrule(lr){5-6} \cmidrule(lr){8-9} \cmidrule(lr){11-12} \cmidrule(lr){14-16}
				$J$                  & $T_0$                   & $r$                  &                         & Bias                 & SD                      &                      & Bias                    & SD   &  & Bias & SD   &  & Bias & SD   & Coverage \\
				\midrule
				\addlinespace[2ex]
				\multicolumn{16}{c}{Panel A: Linear Outcome} \\
        		\addlinespace[1ex]
				25 & 15 & 1 &  & 2.17 & 1.07 &  & 0.99 & 1.17 &  & 1.11 & 1.32 &  & 0.99 & 1.18 & 0.936 \\ 
  50 & 15 & 1 &  & 1.19 & 0.94 &  & 0.77 & 0.93 &  & 0.87 & 1.06 &  & 0.77 & 0.94 & 0.944 \\ 
  25 & 30 & 1 &  & 1.79 & 1.03 &  & 0.90 & 1.09 &  & 0.98 & 1.19 &  & 0.91 & 1.10 & 0.936 \\ 
  50 & 30 & 1 &  & 1.13 & 0.92 &  & 0.75 & 0.92 &  & 0.83 & 1.02 &  & 0.75 & 0.92 & 0.944 \\ 
				\addlinespace[2ex]
				\multicolumn{16}{c}{Panel B: Nonlinear Outcome} \\
        		\addlinespace[1ex]
				25 & 15 & 2 &  & 1.99 & 0.91 &  & 0.94 & 1.02 &  & 0.97 & 1.10 &  & 0.92 & 1.02 & 0.938 \\ 
  50 & 15 & 2 &  & 1.11 & 0.81 &  & 0.77 & 0.84 &  & 0.80 & 0.94 &  & 0.74 & 0.83 & 0.947 \\ 
  25 & 30 & 2 &  & 1.40 & 0.85 &  & 0.87 & 0.96 &  & 0.91 & 1.03 &  & 0.87 & 0.97 & 0.935 \\ 
  50 & 30 & 2 &  & 1.01 & 0.78 &  & 0.69 & 0.80 &  & 0.73 & 0.86 &  & 0.68 & 0.79 & 0.950 \\ 
				\bottomrule
			\end{tabular}
			\begin{tablenotes}
				\item Note: This table compares the bias and SD of the nonlinear SC estimator and three other SC estimators from \cite{abadie2010synthetic}, \cite{doudchenko2017} and \cite{abadie2020penalized} respectively, and reports the coverage probability of the 95\% confidence interval produced by the nonlinear SC method, in different settings that vary in the number of control units $J$, the number of pretreatment periods $T_0$ and the degree of nonlinearity $r$, based on 5000 simulations for each setting.
			\end{tablenotes}
		\end{threeparttable}
	}
\end{center}

\medskip

Table \ref{Ch1_tab_sim} reports the bias and the standard deviation (SD) for each estimator, as well as the coverage probability of the 95\% confidence interval produced by the nonlinear synthetic control estimator in different settings.\footnote{The bias is measured using the average of the absolute biases across the posttreatment periods and the 5000 simulations for each setting. The standard deviation is calculated for each posttreatment period and each set of $(\boldsymbol{X}_i,\boldsymbol{\beta}_t,\boldsymbol{\lambda}_t,\boldsymbol{\mu}_i)$ and then averaged over the posttreatment periods and the 20 sets for each setting. Both the bias and the SD are multiplied by 100 for better presentation.}
With a larger $J$ or $T_0$, the bias and SD become smaller for all the estimators, and the coverage probability of the 95\% confidence interval also improves, in both the linear and nonlinear cases.
Compared with the original estimator, the other estimators have smaller biases across different settings, showing the advantage of more flexible regularisation methods over the non-negativity restriction in reducing the bias of the estimators, whereas the non-negativity restriction has an edge on keeping the SD low.
The bias and SD are on similar levels for ESC and NSC in the linear cases, and are smaller than those of PSC, as ESC and NSC use the $L_2$ regularisation to control the scale of the weights. The advantage of ESC over PSC becomes smaller in the nonlinear cases, since PSC now constructs the synthetic control using closer neighbours so that its bias becomes smaller.
NSC has the smallest bias among the estimators in the nonlinear cases, since it employs both the $L_1$ penalty terms weighted by pairwise matching discrepancies to select closer neighbours, and the $L_2$ regularisation to control the scale of the weights.

\section{Empirical Applications}\label{Ch1_sec_app}

In this section, we first revisit the two empirical applications in \cite{abadie2010synthetic} and \cite{abadie2015comparative} to illustrate the nonlinear synthetic control method, and compare it with the original synthetic control method.
The synthetic controls in both methods are constructed by matching only on the outcomes in all pretreatment periods, since the observed predictors are not essential as long as there is a good fit on the pretreatment outcomes over an extended period of time \citep{botosaru2019role}.
We then move on to the main empirical application of this paper, where we estimate the impact of the 2019 anti-extradition law amendments bill protests in Hong Kong on the city's economy using the nonlinear synthetic control method.

\subsection{California's Tobacco Control Program}

In the first example, we revisit the empirical application in \cite{abadie2010synthetic}, who examine the effect of a large-scale tobacco control program implemented in California in 1988 on the annual per-capita cigarette sales. 

\begin{center}
	\resizebox{12cm}{!}{
		\begin{threeparttable}
			\centering
			\caption{Comparison of Synthetic Control Weights}
			\label{Ch1_tab_cal}
			\begin{tabular}{lp{1.5cm}p{1.5cm}lp{1.5cm}p{1.5cm}}
				\toprule
				State       & OSC Weight & NSC Weight & State          & OSC Weight & NSC Weight \\
				\hline
				Alabama     & 0          & -0.015     & Nevada         & 0.186      & 0.091      \\
				Arkansas    & 0          & -0.057     & New Hampshire  & 0.049      & 0          \\
				Colorado    & 0.03       & 0.119      & New Mexico     & 0          & 0.103      \\
				Connecticut & 0.08       & 0.112      & North Carolina & 0          & 0          \\
				Delaware    & 0          & 0          & North Dakota   & 0          & 0          \\
				Georgia     & 0          & 0          & Ohio           & 0          & 0          \\
				Idaho       & 0          & 0.183      & Oklahoma       & 0          & 0          \\
				Illinois    & 0          & 0.02       & Pennsylvania   & 0          & 0          \\
				Indiana     & 0          & 0          & Rhode Island   & 0          & 0          \\
				Iowa        & 0          & 0.039      & South Carolina & 0          & -0.003     \\
				Kansas      & 0          & 0          & South Dakota   & 0          & 0          \\
				Kentucky    & 0          & 0          & Tennessee      & 0          & -0.071     \\
				Louisiana   & 0          & 0          & Texas          & 0          & 0          \\
				Maine       & 0          & 0          & Utah           & 0.385      & 0.045      \\
				Minnesota   & 0          & 0.027      & Vermont        & 0          & 0          \\
				Mississippi & 0          & -0.007     & Virginia       & 0          & 0          \\
				Missouri    & 0          & 0          & West Virginia  & 0          & 0.083      \\
				Montana     & 0.271      & 0.176      & Wisconsin      & 0          & 0.06       \\
				Nebraska    & 0          & 0.094      & Wyoming        & 0          & 0          \\
				\bottomrule
			\end{tabular}
		\end{threeparttable}
	}
\end{center}

\medskip

Table \ref{Ch1_tab_cal} compares the weights assigned to the other states by the original synthetic control method and the nonlinear synthetic control method ($a^*=0.3$, $b^*=0.7$), respectively.
We see that the weights assigned by the original synthetic control method are concentrated on Utah (0.385), Montana (0.271) and Nevada (0.186), with relatively minor weights on Connecticut, New Hampshire and Colorado. In comparison, the weights in the nonlinear synthetic control method spread out among more states, but are still sparse.

\begin{figure}[!htb]
	\centering
	\begin{subfigure}{.48\textwidth}
		\centering
		\includegraphics[width=\linewidth]{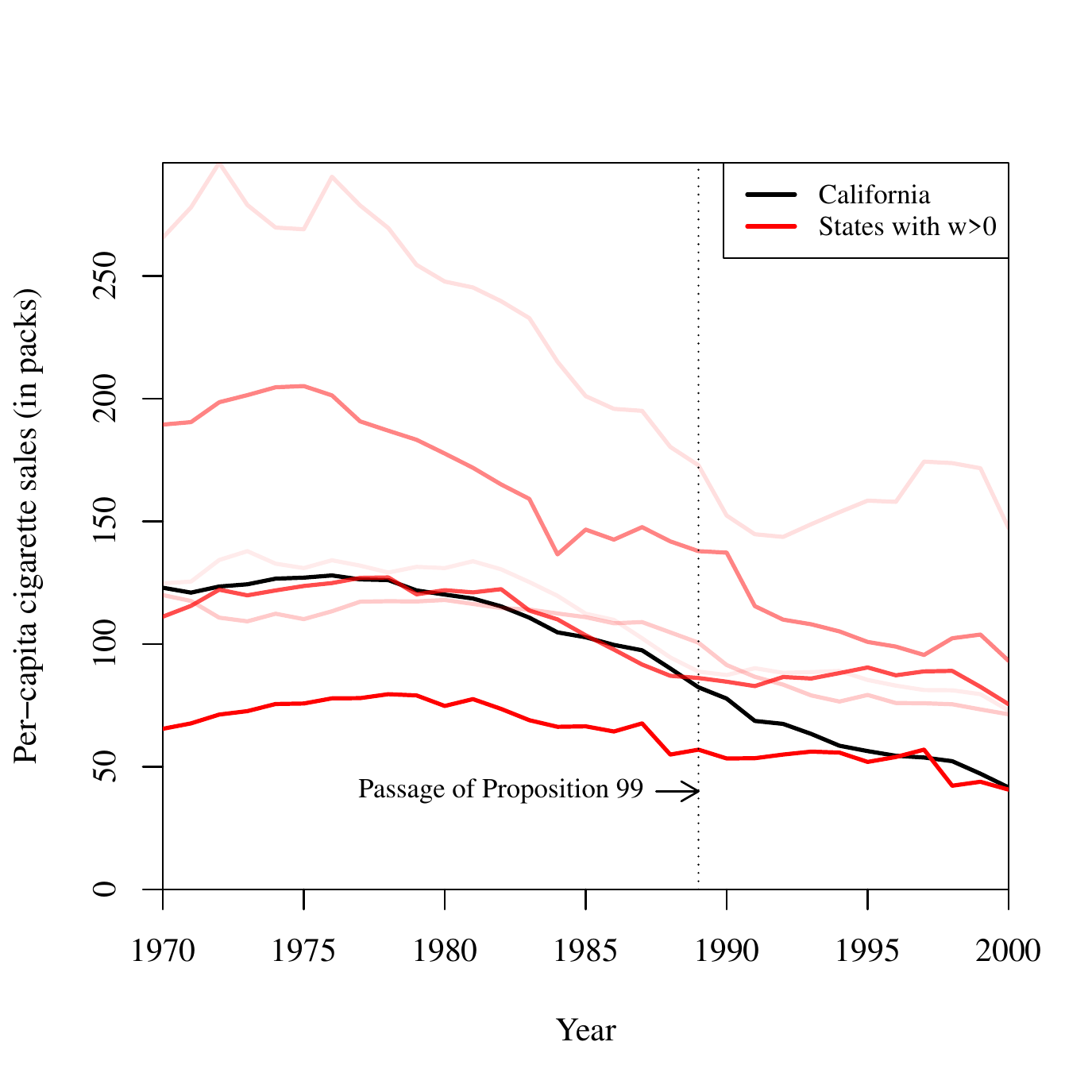}
		\caption{Original SC Method}
		\label{Ch1_fig_cal_com1}
	\end{subfigure}
	~
	\begin{subfigure}{.48\textwidth}
		\centering
		\includegraphics[width=\linewidth]{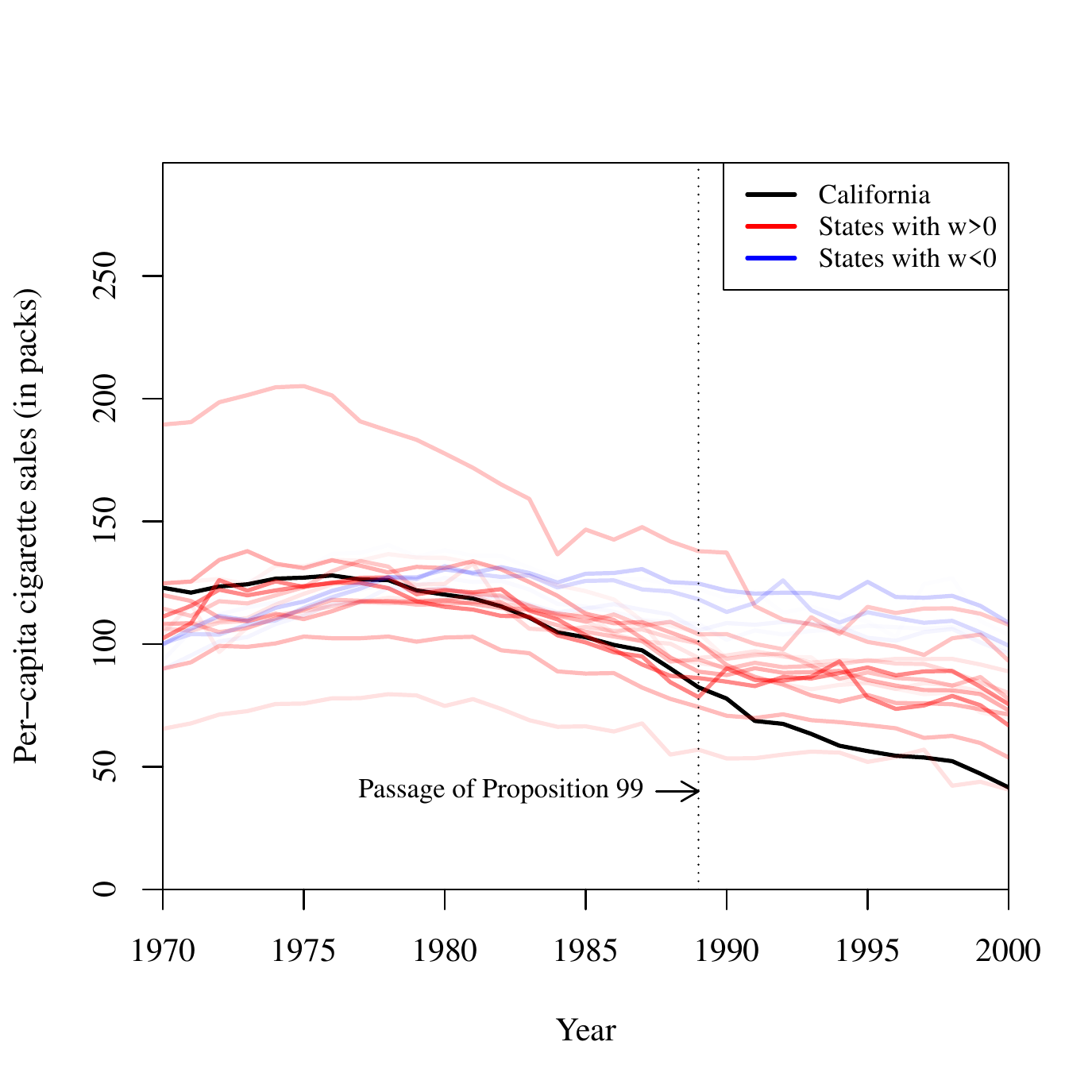}
		\caption{Nonlinear SC Method}
		\label{Ch1_fig_cal_com2}
	\end{subfigure}
	\caption{States Used for Constructing the Synthetic California}
	\label{Ch1_fig_cal_com}
\end{figure}

Figure \ref{Ch1_fig_cal_com} displays the trajectories of the per-capita cigarette sales from 1970 to 2000 for California (black), the states that are assigned positive weights (red), and the states that are assigned negative weights (blue) in the two methods. For comparison, the depth of the colour for the trajectory is scaled by the magnitude of the weight assigned to the state. We see that although the weights in the original synthetic control method are sparse, they may be assigned to control units that are far away from the treated unit, as long as the constructed synthetic control approximates the treated unit well and the weights stay non-negative. 
In comparison, the nonlinear synthetic control method constructs the synthetic control using control units that are close to the treated unit in the presence of nonlinearity.

\begin{figure}[!htb]
	\centering
	\begin{subfigure}{.48\textwidth}
		\centering
		\includegraphics[width=\linewidth]{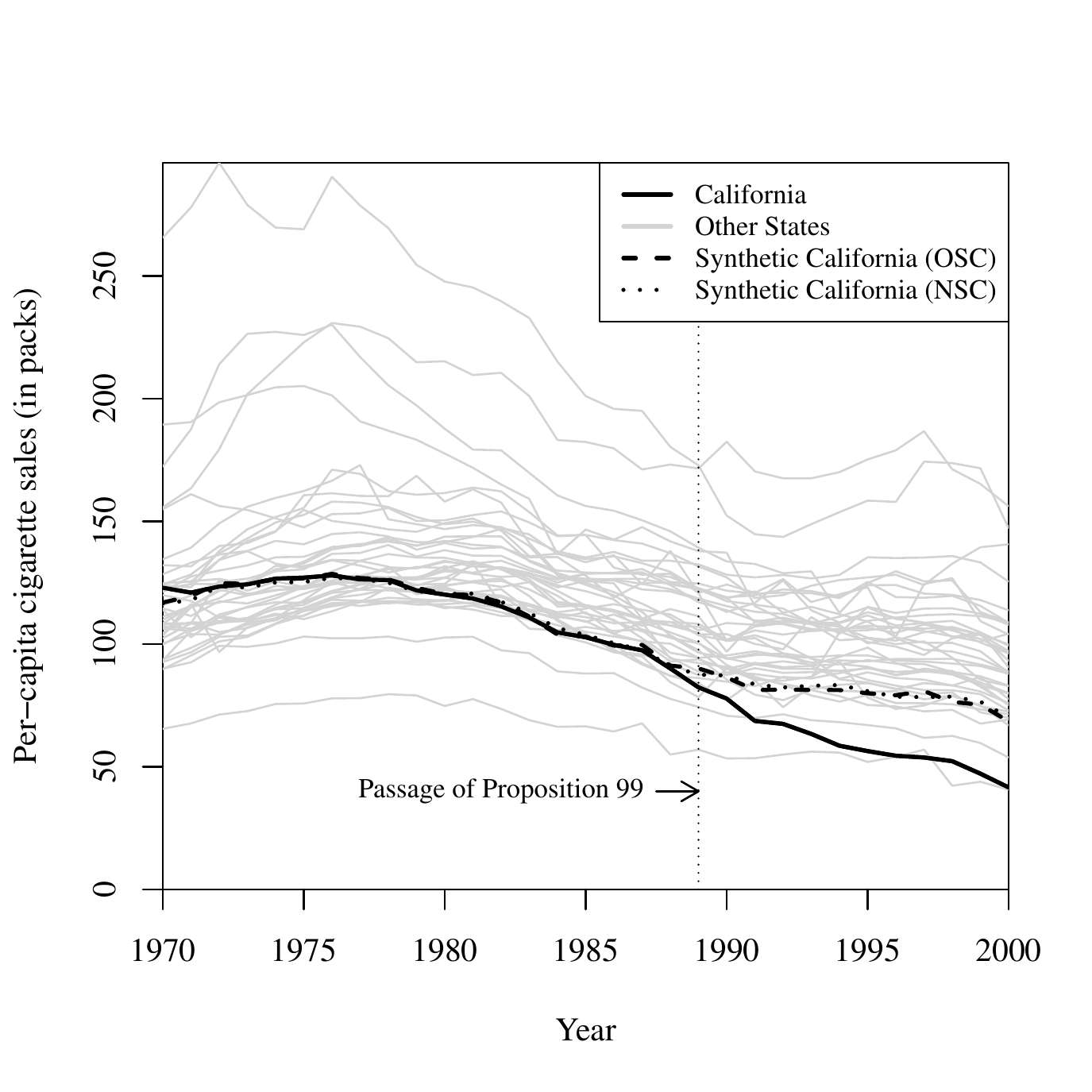}
		\caption{Trajectories of the Outcome}
		\label{Ch1_fig_cal1}
	\end{subfigure}
	~
	\begin{subfigure}{.48\textwidth}
		\centering
		\includegraphics[width=\linewidth]{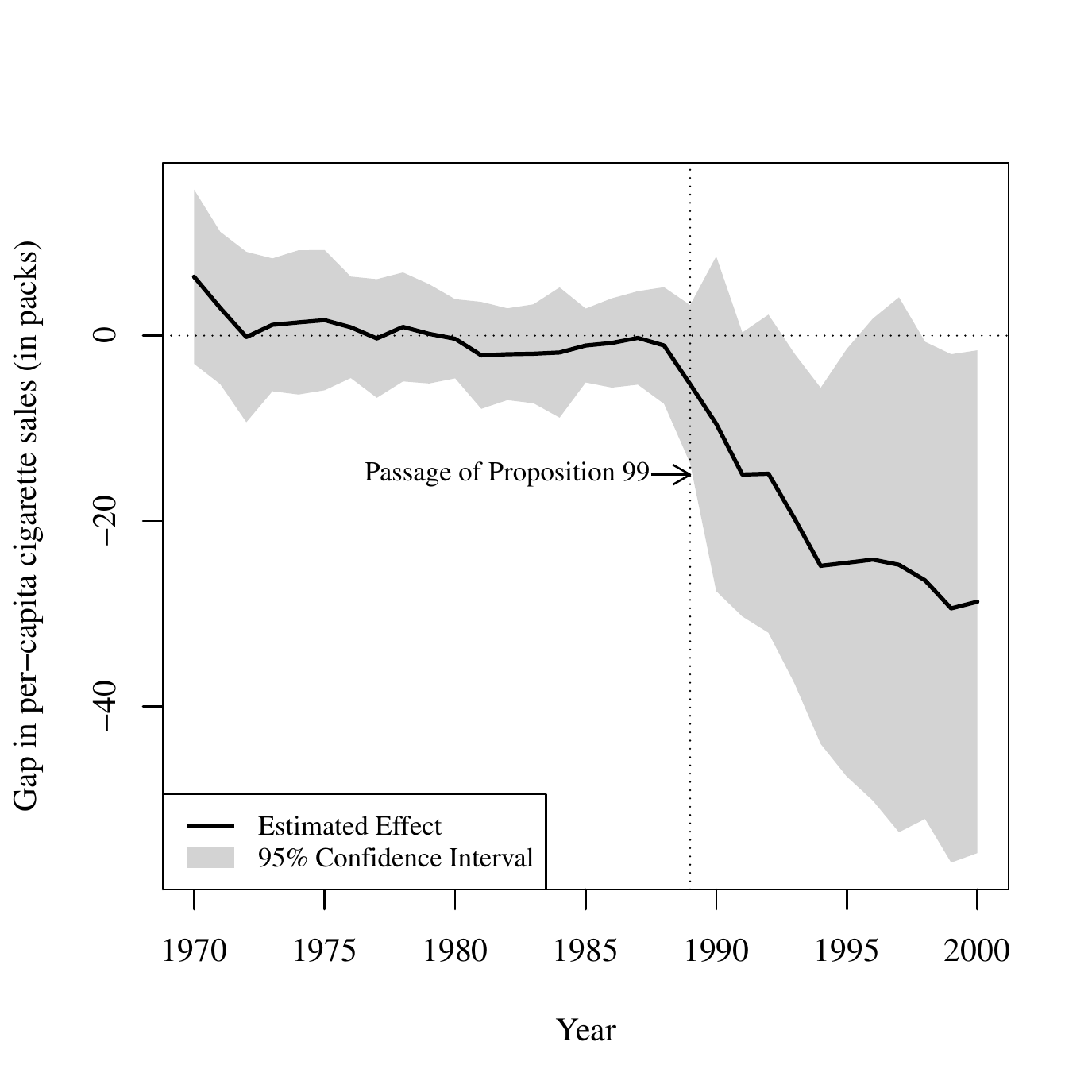}
		\caption{Estimated Treatment Effect}
		\label{Ch1_fig_cal2}
	\end{subfigure}
	\caption{Revisiting the Example of the California Tobacco Law Program}
	\label{Ch1_fig_cal}
\end{figure}

Figure \ref{Ch1_fig_cal1} depicts the trajectories of the per-capita cigarette sales for California (black solid line), the other states (gray solid line), the synthetic California constructed using the original synthetic control method (black dashed line), and the synthetic California constructed using the nonlinear synthetic control method (black dotted line).
We see that the trajectories for the synthetic California from both methods are quite similar, both of which closely follow the trajectory for the real California before 1988, and diverge after the passage of the tobacco control legislation in 1988.
Figure \ref{Ch1_fig_cal2} shows the gap between the trajectories for California and the synthetic California in the nonlinear synthetic control method, which is used to estimate the effect of the tobacco control program on cigarette sales. The result suggests that the tobacco control program reduced per-capita cigarette sales by 9.5 packs in 1990, 24.5 packs in 1995 and 28.7 packs in 2000. The confidence intervals imply that the effect became significant from 1993 onwards (except in 1996 and 1997).

\subsection{1990 German Reunification}

In the second example, we revisit the empirical application in \cite{abadie2015comparative}, which analyse the effect of the 1990 German reunification on West Germany's per-capita GDP.

\begin{center}
	\resizebox{11cm}{!}{
		\begin{threeparttable}
			\centering
			\caption{Comparison of Synthetic Control Weights}
			\label{Ch1_tab_ger}
			\begin{tabular}{lp{1.5cm}p{1.5cm}lp{1.5cm}p{1.5cm}}
				\toprule
				Country   & OSC Weight & NSC Weight & Country     & OSC Weight & NSC Weight \\
				\hline
				Australia & 0          & 0.027      & Netherlands & 0.091      & 0.087      \\
				Austria   & 0.325      & 0.134      & New Zealand & 0          & -0.017     \\
				Belgium   & 0          & 0.101      & Norway      & 0.062      & 0.123      \\
				Denmark   & 0          & 0.058      & Portugal    & 0          & -0.034     \\
				France    & 0          & 0.092      & Spain       & 0          & -0.037     \\
				Greece    & 0.008      & 0.003      & Switzerland & 0.082      & 0.106      \\
				Italy     & 0.062      & 0.096      & UK          & 0.072      & 0.079      \\
				Japan     & 0          & 0.016      & USA         & 0.299      & 0.168      \\
				\bottomrule
			\end{tabular}
		\end{threeparttable}
	}
\end{center}

\medskip

Table \ref{Ch1_tab_ger} compares the weights on the control countries in the original method and the nonlinear synthetic control method ($a^*=0$, $b^*=0.7$).
The weights assigned by the original synthetic control method are concentrated on Austria, USA, the Netherlands, Switzerland, UK, Norway, Italy and Greece, with the weights in descending order. In comparison, the weights assigned by the nonlinear synthetic control method spread out among all countries, with positive weights on countries that are close to West Germany in terms of the outcomes, and negative weights on countries that are farther away.
This can also be seen from Figure \ref{Ch1_fig_ger_com}, which depicts the trajectories of GDP per capita for countries used for constructing the synthetic West Germany in the two methods.

\begin{figure}[!htb]
	\centering
	\begin{subfigure}{.48\textwidth}
		\centering
		\includegraphics[width=\linewidth]{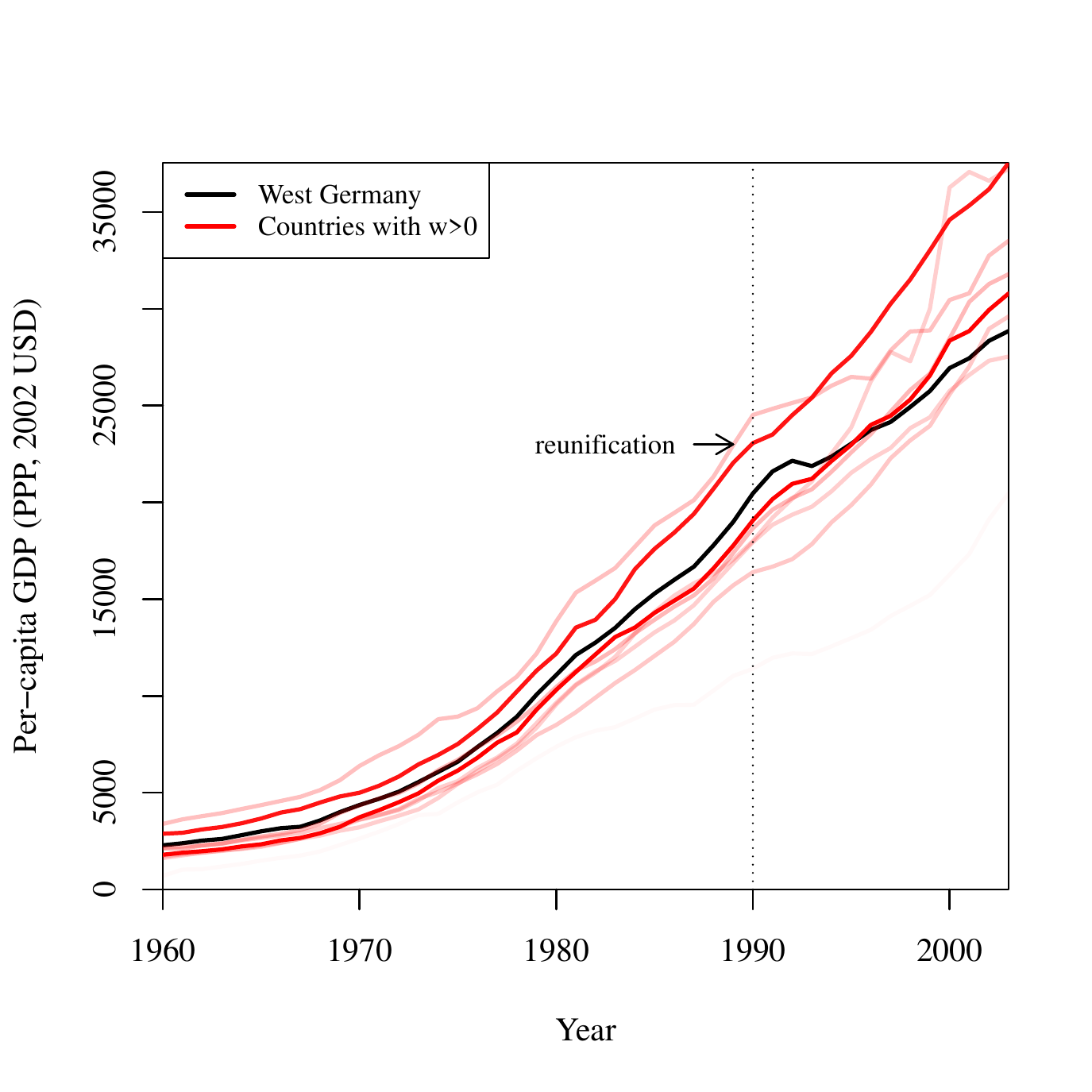}
		\caption{Original SC Method}
		\label{Ch1_fig_ger_com1}
	\end{subfigure}
	~
	\begin{subfigure}{.48\textwidth}
		\centering
		\includegraphics[width=\linewidth]{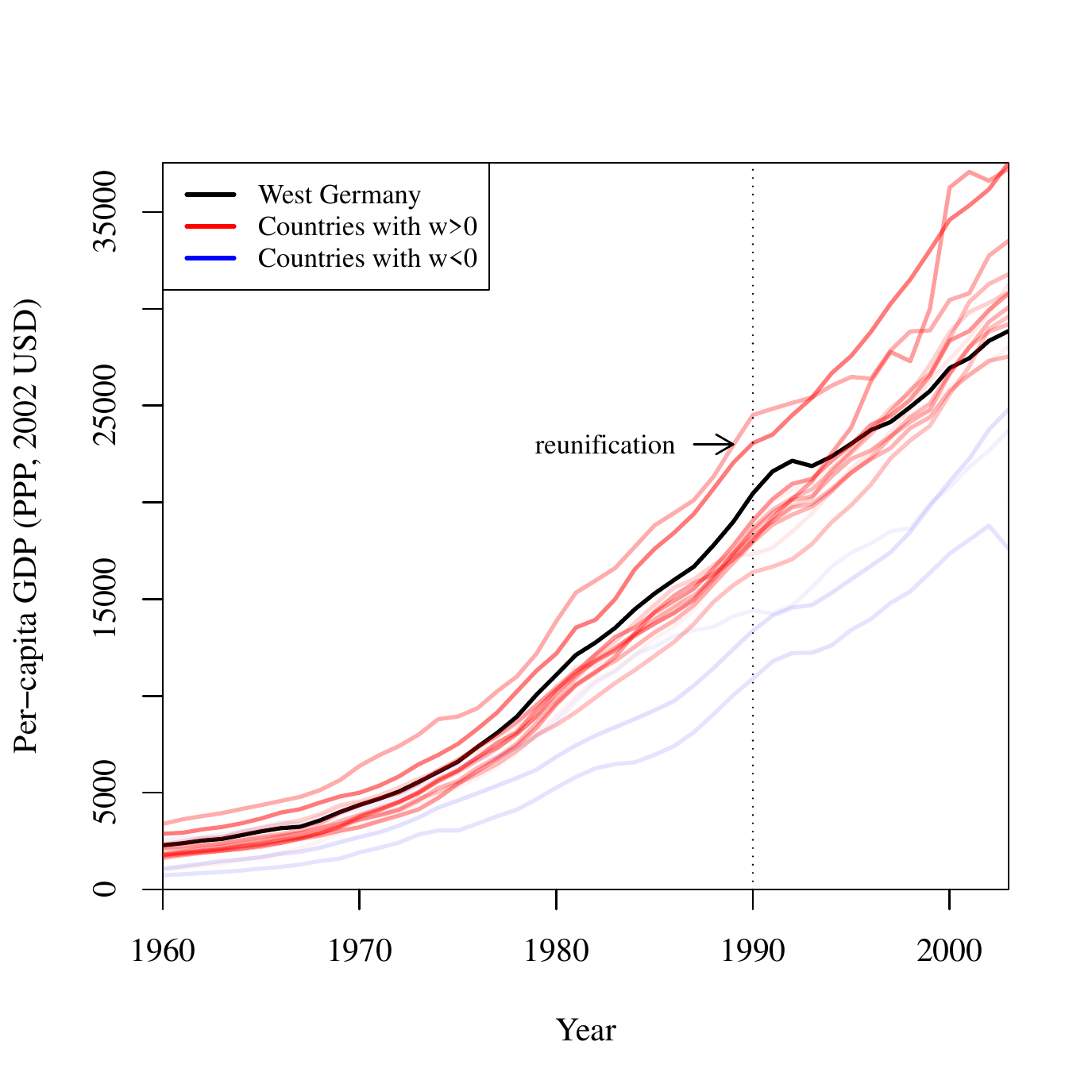}
		\caption{Nonlinear SC Method}
		\label{Ch1_fig_ger_com2}
	\end{subfigure}
	\caption{Countries Used for Constructing the Synthetic West Germany}
	\label{Ch1_fig_ger_com}
\end{figure}

\begin{figure}[!htb]
	\centering
	\begin{subfigure}{.48\textwidth}
		\centering
		\includegraphics[width=\linewidth]{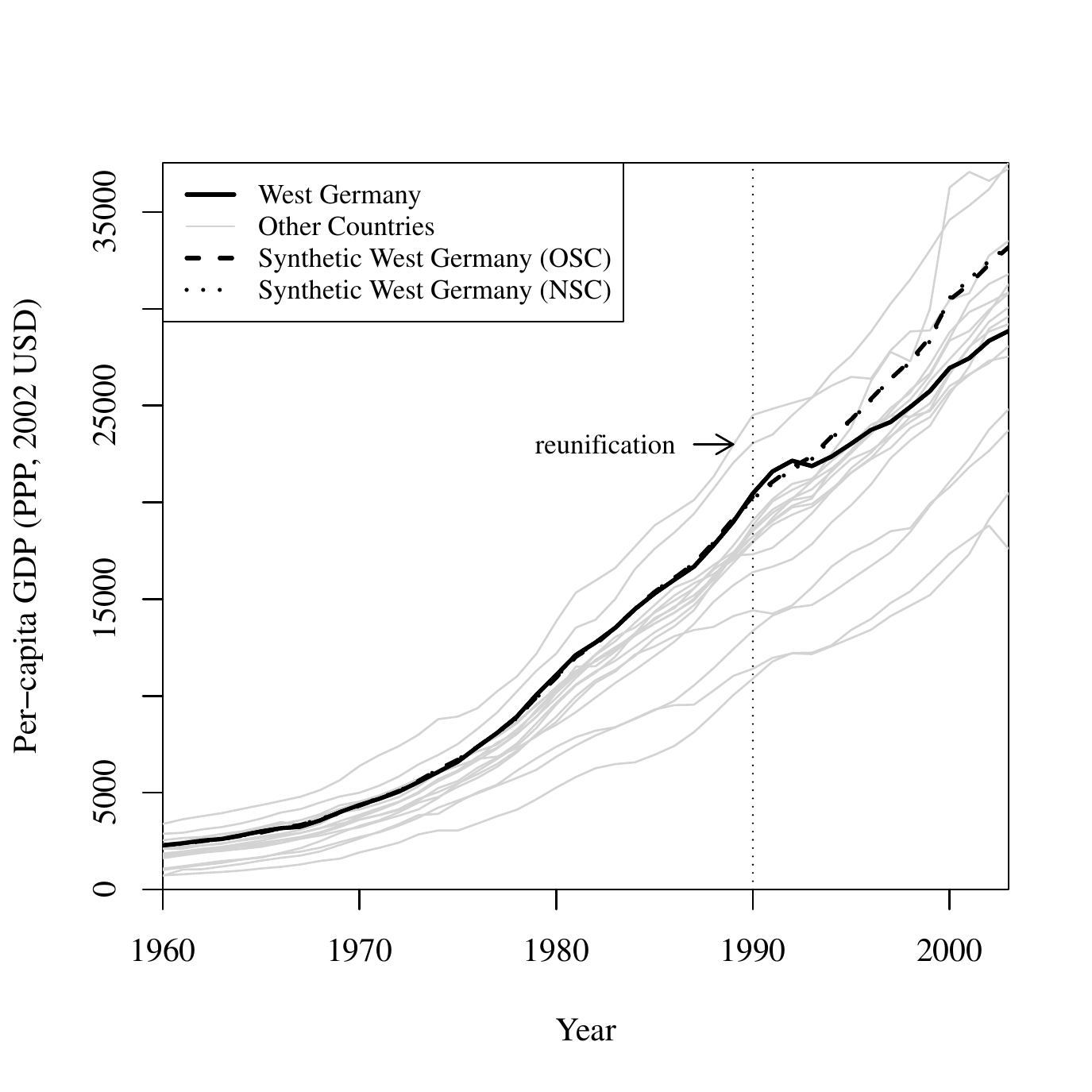}
		\caption{Trajectories of the Outcome}
		\label{Ch1_fig_ger1}
	\end{subfigure}
	~
	\begin{subfigure}{.48\textwidth}
		\centering
		\includegraphics[width=\linewidth]{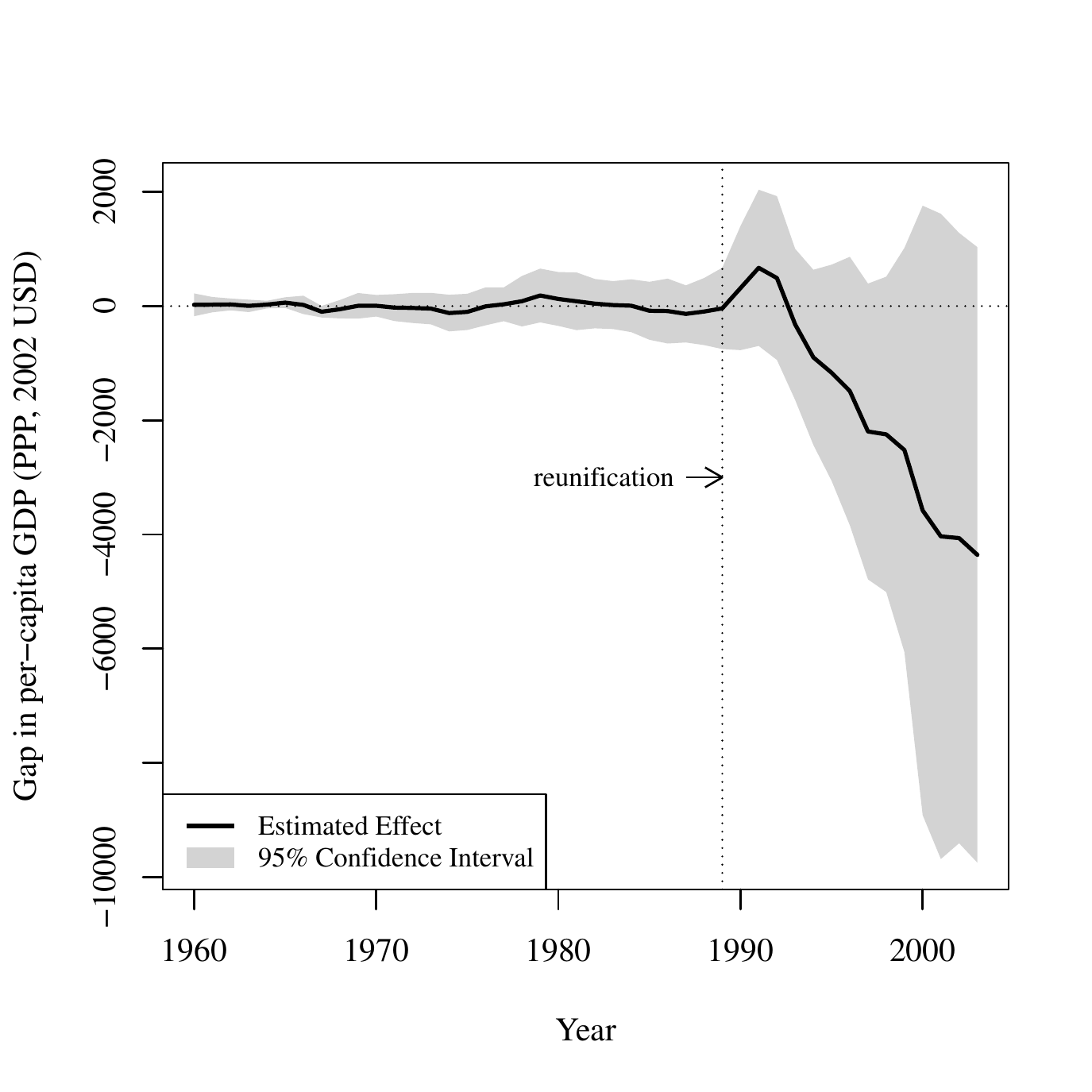}
		\caption{Estimated Effect}
		\label{Ch1_fig_ger2}
	\end{subfigure}
	\caption{Revisiting the Example of German Reunification}
	\label{Ch1_fig_ger}
\end{figure}

Figure \ref{Ch1_fig_ger1} depicts the trajectories of GDP per capita for West Germany (black solid line), the other countries (gray solid lines), the synthetic West Germany constructed using the original synthetic control method (black dashed line), and the synthetic West Germany constructed using the nonlinear synthetic control method (black dotted line). Despite differences in the weights, the trajectories for the two synthetic West Germanies are virtually the same, both of which follow the trajectory for the real West Germany closely before the German reunification in 1990, and diverge afterwards.
The gap between the trajectories for West Germany and the synthetic West Germany in the nonlinear synthetic control method indicates that the reunification reduced the per-capita GDP in West Germany by 1166 USD in 1995, 2520 USD in 1999, and 4356 USD in 2003. 
However, we fail to reject the null hypothesis that German reunification had no effect on West Germany's economy in any particular period.

\subsection{The Economic Impact of the 2019 Hong Kong Protests}

\subsubsection{Background}

In 2018, a young couple from Hong Kong, the 20-year-old woman, Poon Hiu-wing, and her 19-year-old boyfriend, Chan Tong-kai, travelled to Taiwan as tourists. Following a quarrel in the hotel room, Chan strangled Poon, took her valuables and fled back to Hong Kong \citep{bbc2019}.
Since the murder took place in Taiwan where the Hong Kong authorities had no jurisdiction, they could only charge Chan with money laundering but not homicide. Nor could they surrender Chan to Taiwan, as there was no extradition treaty or one-off surrender agreement between Hong Kong and Taiwan.\footnote{The term `surrender' is used formally in place of `extradition' to reflect that Hong Kong is part of China and does not have sovereign status.}

To fill this legal loophole, the Hong Kong government proposed amendments to the existing extradition law (formally, the Fugitive Offenders Ordinance, Cap. 503) in February 2019, to allow case-based surrenders of fugitive offenders to jurisdictions apart from the twenty with which the city already had extradition treaties \citep{legco2019}.\footnote{Amendments were also proposed to the Mutual Legal Assistance in Criminal Matters Ordinance, Cap. 525, which were less controversial.}\textsuperscript{,}\footnote{Some do not consider the limitation that excludes the rest of China from both ordinances a loophole, but rather a deliberate restriction to protect Hong Kong's legal system \citep{HKBA2019}. However, since Hong Kong does not have the authority to enter into an extradition treaty or one-off surrender agreement only with Taiwan but not the other parts of China, this allows fugitives as Chan in the Taiwan homicide case to evade prosecution, and thus is effectively a loophole in this regard.
Some lawmakers also advocated adding a sunset clause to the amendments, which would see the amendments expire after the resolution of the Taiwan case \citep{standard201904}. This was rejected by the government, who reiterated that the purpose of the amendments was not only to resolve the Taiwan case, but also to improve the existing arrangement for the surrender of fugitive offenders \citep{HKSARG201904}.}
While the proposed amendments would enable Hong Kong to surrender Chan to Taiwan, the inclusion of mainland China raised concerns among residents from different walks of life in the city that civil liberties would be infringed upon, given previous incidents of Hong Kong residents being abducted to the mainland for trial \citep{AP2016}. The open support of the amendments from several central government officials contributed to the rising anxieties \citep{Reuters2019}.

The amendments bill sparked a series of protests, which started as peaceful demonstrations in March and April, and eventually escalated into violence from June, when the government pushed for a speedy second reading of the bill to ensure its passage before the release of Chan from prison on money laundering charges.
Hundreds of thousands joined the protests, and clashes broke out between the protesters and the police, with radical protesters throwing bricks dug up from the pavement, iron bars disassembled from the roadside railings, and later petrol bombs at the police, who responded with pepper spray, tear gas, and rubber bullets \citep{SCMP201906,mingpao2019}.
Unlike the protests in previous years, the moderate protesters refused to split with the radical protesters this time, as most of them believed that ``peaceful assembly should combine with confrontational actions to maximise the impact of protests'', and that ``radical tactics were understandable when the government refuses to listen'', despite other peaceful avenues such as strikes \citep{Yuen2019}.

The suspension of the bill by the government on 15 June did not quiet, but rather boosted the morale of the protesters, who raised more demands including the full withdrawal of the bill, the retraction of the characterisation of the protests as ``riots'', the release and exoneration of the arrested protesters, the establishment of an independent commission of inquiry into police brutality, the resignation of Carrie Lam as chief executive, and the universal suffrage for the Legislative Council and the chief executive elections, pushing the protests to a non-resolvable end.
Led by pro-independence activists who had been at the forefront of the protests from the beginning, the protests also evolved from aiming against the amendments bill to be against China or the Chinese government, challenging the ``one country, two systems'' principle \citep{HKFP2019}.
The oftentimes violent protests persisted through the next few months, which saw the police headquarters besieged, the Legislative Council stormed, the Liaison Office of the Central People's Government attacked, the national emblem and flag of China desecrated, the international airport occupied, railway stations and shops vandalised, and several universities sieged by the protesters \citep{HKFP20190621,CNN201907,HKFP20190722,Huff2019,cnn201909,ABC201911,BBC201911,HKFP2021}. Up till 14 April 2020, 8001 protesters were arrested, among whom 41\% were students, with 60\% being university students and 40\% being secondary school students \citep{chinanews2020}.

On 30 June 2020, the Standing Committee of the National People's Congress of China passed the Hong Kong national security law, which was to be enacted by the city on its own, but which the city had failed to accomplish since its return to China in 1997.\footnote{Article 23 of Hong Kong's Basic Law: ``The Hong Kong Special Administrative Region shall enact laws on its own to prohibit any act of treason, secession, sedition, subversion against the Central People's Government, or theft of state secrets, to prohibit foreign political organizations or bodies from conducting political activities in the Region, and to prohibit political organizations or bodies of the Region from establishing ties with foreign political organizations or bodies.''\citep{basiclaw1997}}
Within hours, several pro-independence organisations announced the decision to disband and cease all operations \citep{aljazeera2020}.
Subsequently, a dozen Legislative Council candidates were disqualified and several pro-independence activists were arrested under the national security law, bringing the year-long unrest in the city to a halt.

\subsubsection{Data}

To estimate the economic impact of the 2019 anti-extradition law amendments bill protests in Hong Kong, we compare the quarterly GDP per capita of Hong Kong and the synthetic Hong Kong constructed using 48 other major economies listed in Table \ref{Ch1_tab_weights}. The quarterly GDP per capita is measured in chained (2015) U.S. dollars and is Purchasing Power Parity (PPP) and seasonally adjusted.
The treatment assignment period is the first quarter of 2019, during which the government proposed the amendments bill to the existing extradition law and the first round of protests was triggered.
The window of observation is from the first quarter of 2011 to the fourth quarter of 2020, which is chosen to provide enough pretreatment periods, and in the meantime, to avoid potential structural breaks, e.g., due to the 2008 global financial crisis, over longer periods of time \citep{abadie2021jel}.

The data for Hong Kong, Taiwan and Singapore are obtained from the respective government statistics websites. Purchasing power parities and the data for the other economies are obtained from OECD Stat, among which the seasonally adjusted GDP series for China is not available and is computed using the growth rate of the seasonally adjusted GDP. And the population for all economies is obtained from the United Nations. More detailed data sources are provided in Appendix \ref{Ch1_sec_data}.

\begin{figure}[!htb]
	\centering
	\includegraphics[width=.5\linewidth]{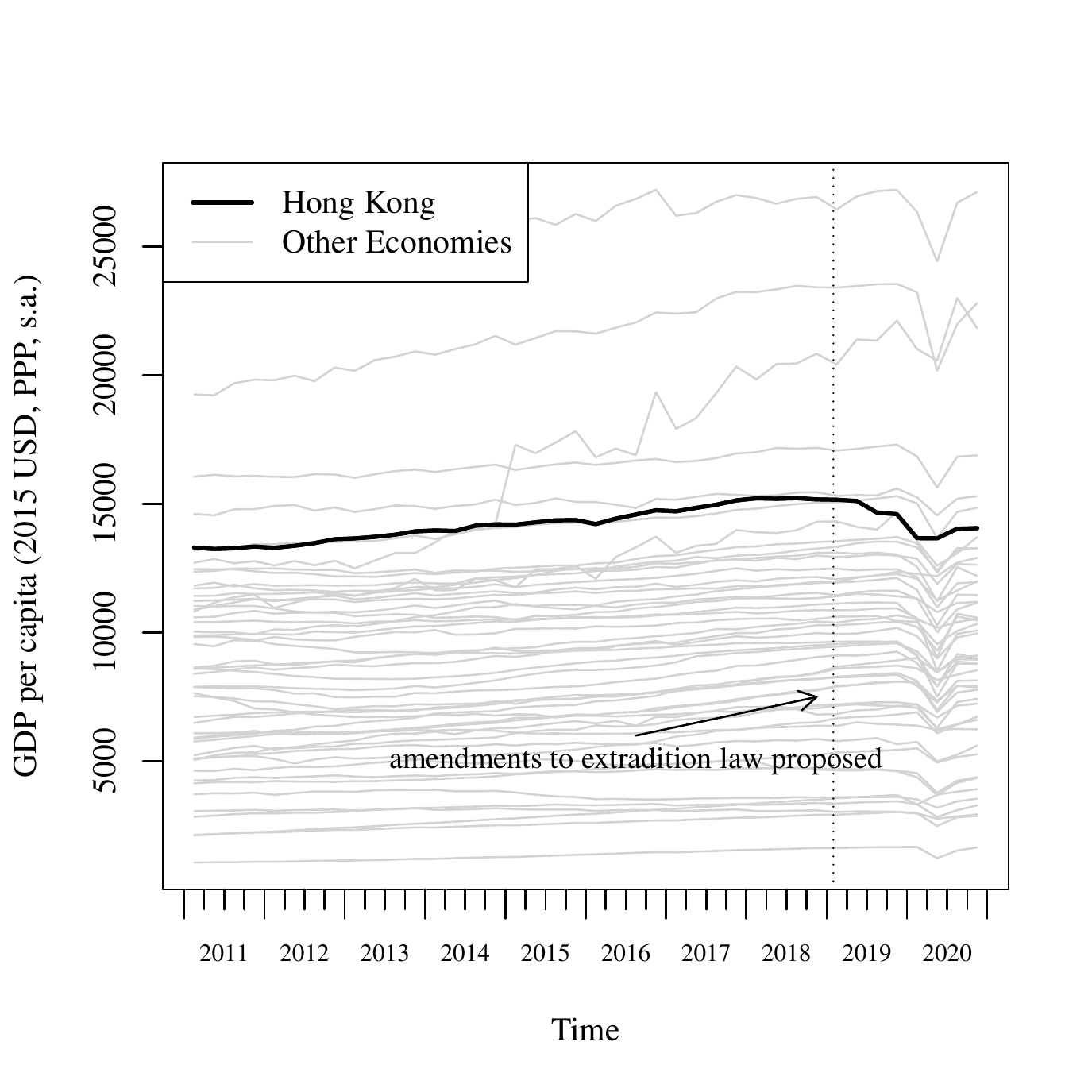}
	\caption{Trajectories of GDP per capita}
	\label{Ch1_fig_SC_des}
\end{figure}

Figure \ref{Ch1_fig_SC_des} visualises the trajectories of the quarterly GDP per capita for Hong Kong (black line) and the other economies in the sample (gray lines). We see that Hong Kong had one of the highest GDP per capita in the sample, which would be better approximated by the other economies without the non-negativity restriction. Most economies in the sample enjoyed steady growth in GDP per capita from 2011 until the first or second quarter of 2020, when almost all economies were severely hit by the COVID-19 pandemic. Most economies then had strong rebounds in the third quarter. In contrast, the decline of the GDP per capita in Hong Kong started from early 2019, coinciding with the onset of the protests, and persisted through the outbreak of the pandemic, while the recovery from the third quarter of 2020 was only mild.

This preliminary comparison clearly points to a detrimental effect of the protests on Hong Kong's economy. To estimate the effect more accurately, we construct a synthetic Hong Kong that closely tracks the GDP per capita of Hong Kong before the protests in 2019 using the other economies in the sample, which presumably is also close to Hong Kong in the underlying predictors and thus can be used to predict the counterfactual outcome of Hong Kong. We can then estimate the economic impact of the 2019 protests using the difference in GDP per capita between the synthetic Hong Kong and the real Hong Kong. The weights assigned to the other economies for constructing the synthetic Hong Kong are determined by the nonlinear synthetic control method to ensure a small bias of the estimator given the data.

\subsubsection{Results}

\begin{center}
	\resizebox{15cm}{!}{
		\begin{threeparttable}
			\centering
			\caption{Synthetic Control Weights}\label{Ch1_tab_weights}
			\begin{tabular}{lp{1.2cm}p{1.2cm}lp{1.2cm}p{1.2cm}lp{1.2cm}p{1.2cm}}
				\toprule
				Location         & OSC Weight & NSC Weight & Location    & OSC Weight & NSC Weight & Location       & OSC Weight & NSC Weight \\
				\hline
				Argentina & 0.12 & 0 & Greece & 0 & 0 & Norway & 0 & 0.01 \\ 
				Australia & 0 & 0 & Hungary & 0 & 0 & Poland & 0 & 0 \\ 
				Austria & 0 & 0 & Iceland & 0.03 & 0.06 & Portugal & 0 & 0 \\ 
				Belgium & 0 & 0 & India & 0 & 0 & Romania & 0 & 0 \\ 
				Brazil & 0 & 0 & Indonesia & 0 & 0 & Russia & 0 & 0 \\ 
				Bulgaria & 0 & 0 & Ireland & 0 & -0.01 & Singapore & 0.24 & 0.17 \\ 
				Canada & 0 & 0.01 & Israel & 0 & 0 & Slovakia & 0 & 0 \\ 
				Chile & 0 & 0 & Italy & 0 & 0 & Slovenia & 0 & -0.02 \\ 
				China (Mainland) & 0 & 0 & Japan & 0 & 0.07 & South Africa & 0 & 0 \\ 
				Colombia & 0 & 0 & South Korea & 0 & 0 & Spain & 0 & 0 \\ 
				Czechia & 0 & 0 & Latvia & 0 & 0 & Sweden & 0 & 0 \\ 
				Denmark & 0.02 & 0.05 & Lithuania & 0 & 0 & Switzerland & 0.19 & 0.15 \\ 
				Estonia & 0 & 0 & Luxembourg & 0.04 & 0.03 & Taiwan & 0.3 & 0.24 \\ 
				Finland & 0 & 0 & Mexico & 0 & 0 & Turkey & 0.06 & 0.07 \\ 
				France & 0 & 0 & Netherlands & 0 & 0 & United Kingdom & 0 & 0.02 \\ 
				Germany & 0 & 0.02 & New Zealand & 0 & 0 & United States & 0 & 0.13 \\ 
				\bottomrule
			\end{tabular}
		\end{threeparttable}
	}
\end{center}

\medskip

Table \ref{Ch1_tab_weights} displays the weights assigned to the other economies by the original method and the nonlinear synthetic control method ($a^*=0.6$ and $b^*=0.6$). 
The synthetic Hong Kong in the original method is constructed as a weighted average of Taiwan, Singapore, Switzerland, Argentina, Turkey, Luxembourg, Iceland and Denmark, with the weights in descending order. All the other economies receive zero weight. The weights in the nonlinear synthetic control method are similar to those in the original method, with the noticeable difference that instead of assigning significant weight to Argentina, the nonlinear synthetic control method assigns weight to the United States, which is more similar to Hong Kong in terms of the outcome. This is also reflected in Figure \ref{Ch1_fig_HK_com}. Although the weights in the original method are more sparse, the weights in the nonlinear synthetic control method are more concentrated on economies that are closer to Hong Kong.

\begin{figure}[!htb]
	\centering
	\begin{subfigure}{.48\textwidth}
		\centering
		\includegraphics[width=\linewidth]{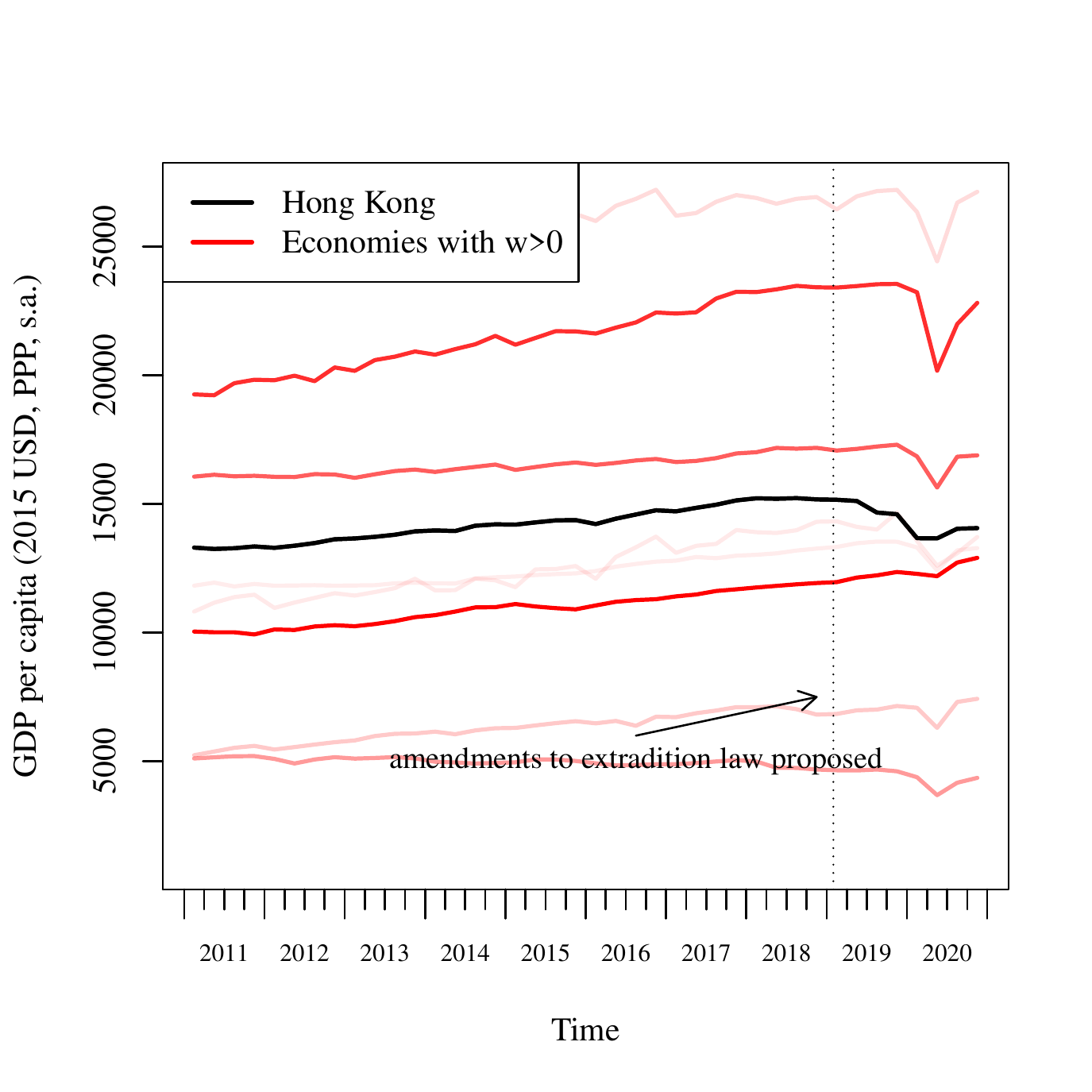}
		\caption{Original SC Method}
		\label{Ch1_fig_HK_com1}
	\end{subfigure}
	~
	\begin{subfigure}{.48\textwidth}
		\centering
		\includegraphics[width=\linewidth]{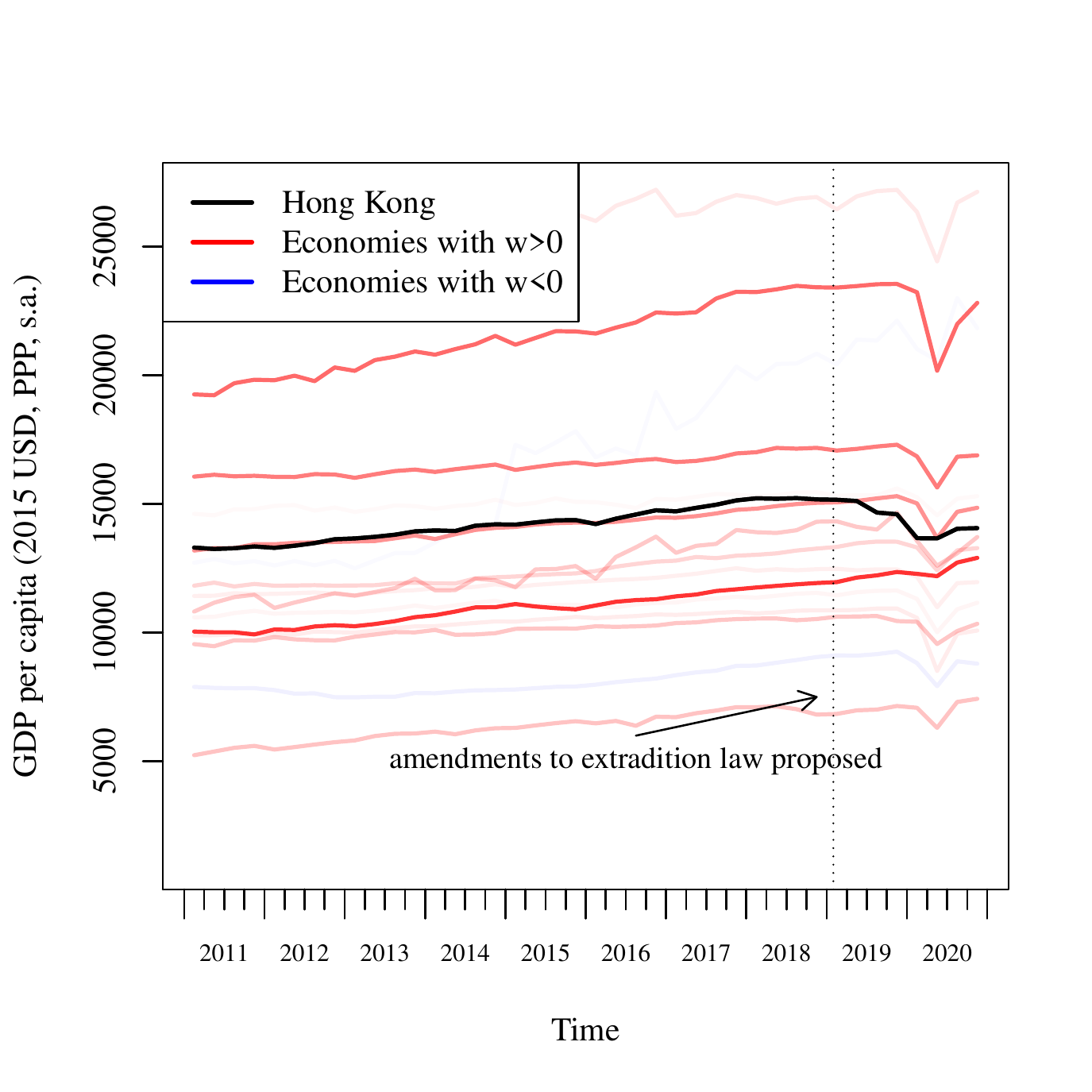}
		\caption{Nonlinear SC Method}
		\label{Ch1_fig_HK_com2}
	\end{subfigure}
	\caption{Economies Used for Constructing the Synthetic Hong Kong}
	\label{Ch1_fig_HK_com}
\end{figure}

\begin{figure}[!htb]
	\centering
	\begin{subfigure}{.48\textwidth}
		\centering
		\includegraphics[width=\linewidth]{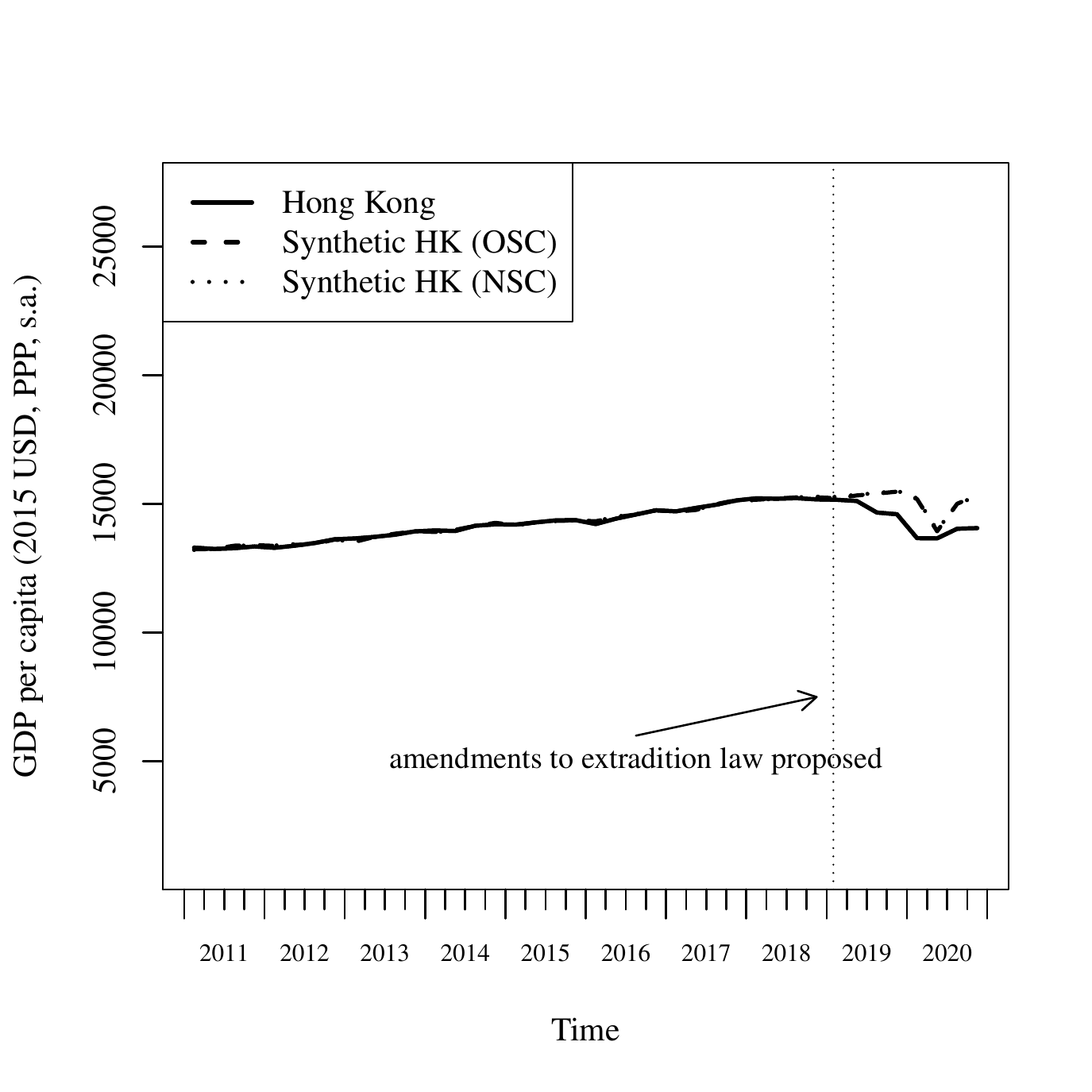}
		\caption{Trajectories for HK and synthetic HK}
		\label{Ch1_fig_HK_SC}
	\end{subfigure}
	~
	\begin{subfigure}{.48\textwidth}
		\centering
		\includegraphics[width=\linewidth]{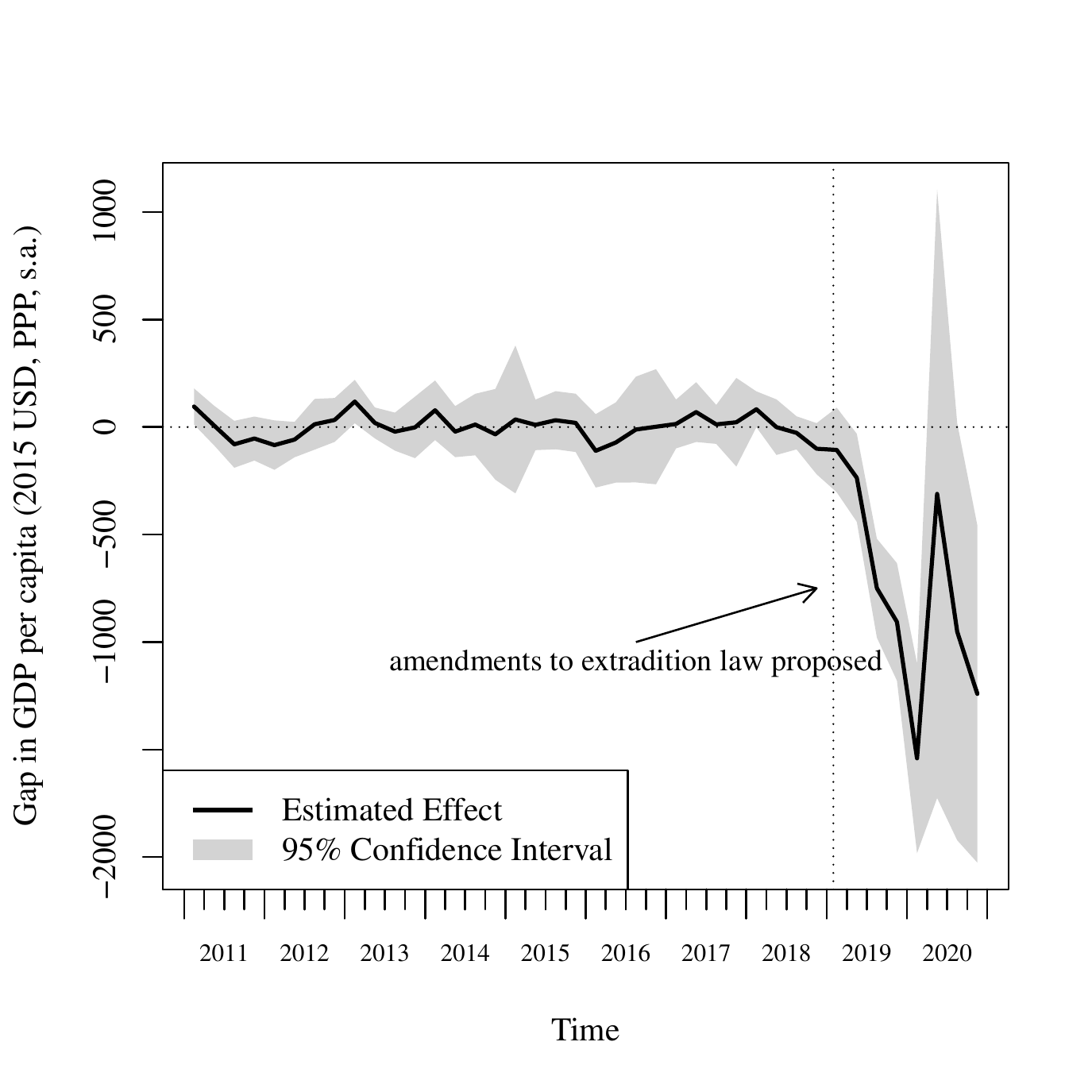}
		\caption{Estimated effect of protests}
		\label{Ch1_fig_HK_TE}
	\end{subfigure}
	\caption{Effect of Anti-Extradition Law Protests in Hong Kong}
	\label{Ch1_fig_HK}
\end{figure}

Figure \ref{Ch1_fig_HK_SC} displays the trajectories of quarterly GDP per capita for Hong Kong (black solid line), the synthetic Hong Kong constructed using the original synthetic control method (black dashed line), and the synthetic Hong Kong constructed using the nonlinear synthetic control method (black dotted line).
Despite the differences in the weights, the trajectories of the synthetic Hong Kong using the two methods are very similar and track the trajectory of Hong Kong very closely before the proposal of the amendments bill in the first quarter of 2019, and begin to diverge immediately afterwards.
This also indicates that the results are not sensitive to the choice of the tuning parameters.
Figure \ref{Ch1_fig_HK_TE} depicts the gap between the trajectories for Hong Kong and the synthetic Hong Kong constructed using the nonlinear synthetic control method, which is used to estimate the economic impact of the anti-extradition law amendments bill protests, as well as the 95\% confidence intervals.\footnote{Note that the confidence intervals in 2019 seem narrow only due to the slope. The average width of the confidence intervals is 274 in the pretreatment periods, and 452 in 2019.} The results suggest that the protests had a negative impact on Hong Kong's economy from the second quarter of 2019. The magnitude of the impact grew rapidly and reached its peak in the first quarter of 2020, when the GDP per capita in Hong Kong was 1540.76 USD or 11.27\% lower than what it would be if there were no protests. To put it into perspective, this magnitude exceeds the peak-to-trough decline in quarterly GDP per capita in Hong Kong during the previous two financial crises, which was 11.14\% in the 1997 Asian financial crisis, and 8.08\% in the 2008 global financial crisis, calculated using the same data series.
The impact was no longer significant in the second and third quarters of 2020, but became significant again in the fourth quarter, with the quarterly GDP per capita 8.8\% lower than its counterfactual level due to the slow recovery of the economy in Hong Kong. 

Note that the above results rely on the assumption that there is no spill-over effect. If the protests in Hong Kong benefited competing economies such as Taiwan and Singapore by driving capital and labour to those economies, then the economic impact of the protests would be overestimated. On the other hand, if the protests in Hong Kong had negative spill-over effects on the economies that receive nonnegative weights, e.g., by damaging the economic cooperation, then the impact of the protests would be underestimated.
Our results may also be confounded by the COVID-19 pandemic, which is a separate treatment from the protests. Although this treatment affected all economies, the magnitudes of the effects may be different. For example, there was virtually no further decline in the battered economy of Hong Kong in 2020, while the COVID-19 pandemic devastated almost all the other economies. The pandemic may even have benefitted Hong Kong's economy by restricting the protests. Thus, the results in 2020 should be interpreted with this in mind.

\subsubsection{Robustness Checks}

The synthetic Hong Kong is constructed by matching on the quarterly GDP per capita of Hong Kong from 2011 to 2019. Although unlikely, people might have anticipated a proposal to amend the extradition law and the turbulence that might follow, after the Taiwan homicide case took place in 2018.
To get rid of the potential anticipatory effect, we backdate the treatment to the first quarter of 2017, two years before the real treatment took place and one year before the Taiwan homicide case that triggered the government proposal of the amendments bill, and construct the synthetic Hong Kong by matching on the quarterly GDP per capita before 2017, to see if the results are sensitive to the choice of the treatment date.
This exercise also allows us to examine the ability of the synthetic Hong Kong to replicate the quarterly GDP per capita of Hong Kong in the absence of the treatment. If we were to find a large gap between the trajectories of Hong Kong and the synthetic Hong Kong after the placebo treatment in 2017 and before the real treatment in 2019, then it would undermine the credibility of the previous results.

\begin{figure}[!htb]
	\centering
	\begin{subfigure}{.48\textwidth}
		\centering
		\includegraphics[width=\linewidth]{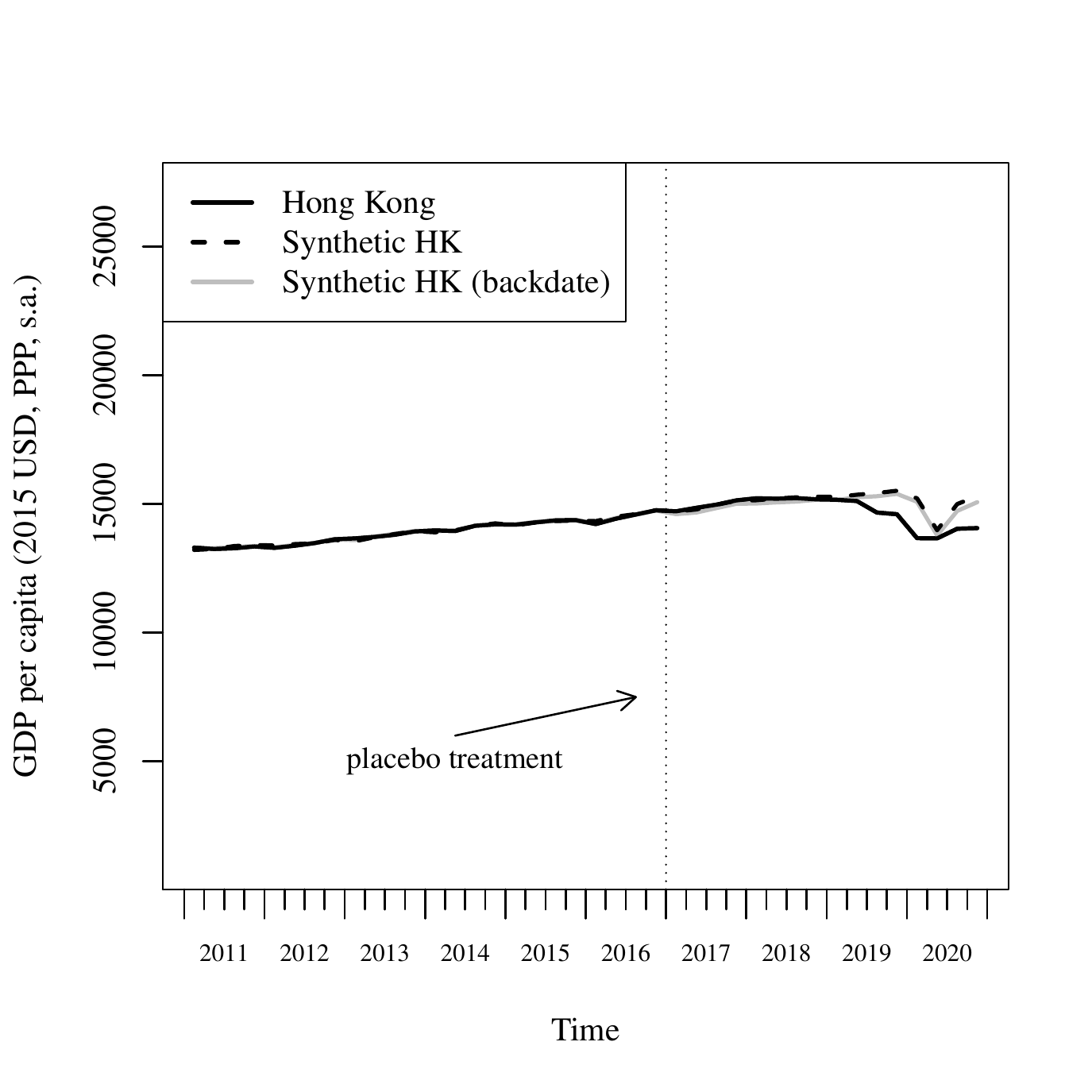}
		\caption{Trajectories for HK and synthetic HK}
		\label{Ch1_fig_SC_BD}
	\end{subfigure}
	~
	\begin{subfigure}{.48\textwidth}
		\centering
		\includegraphics[width=\linewidth]{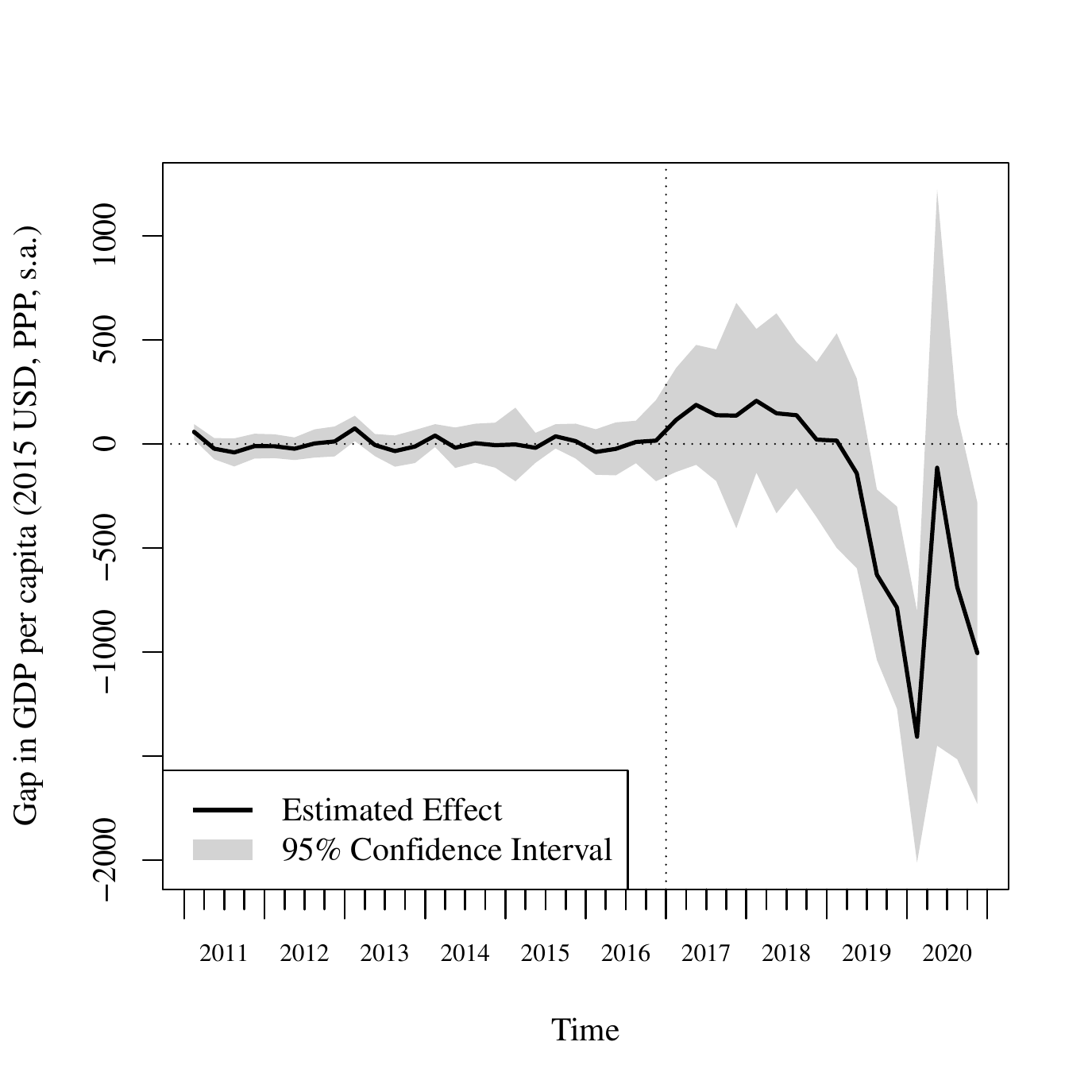}
		\caption{Estimated effect of protests}
		\label{Ch1_fig_TE_BD}
	\end{subfigure}
	\caption{Backdating the Treatment to 2017}
	\label{Ch1_fig_BD}
\end{figure}

The results of the backdating exercise are presented in Figure \ref{Ch1_fig_BD}, which turn out to be very similar to the previous results. We find no significant placebo treatment effect as the trajectory of the newly constructed synthetic Hong Kong (gray solid line) follows that of Hong Kong closely not only before the placebo treatment in the first quarter of 2017, but also all the way through 2017 and 2018, and begins to diverge immediately after the real treatment took place in early 2019.
The magnitude of the estimated treatment effect is close to the previous result, and the 95\% confidence intervals in the backdating exercise similarly suggest that the treatment effect becomes significant from the third quarter of 2019 onwards, except in the second and third quarter of 2020.
The fact the the synthetic Hong Kong is able to reproduce the GDP per capita of Hong Kong in the absence of the real treatment shows the credibility of the synthetic control estimator. And the emergence of the estimated effect shortly after the real treatment provides confidence that the results are driven by a true detrimental impact of the protests.

\begin{figure}[!htb]
	\centering
	\begin{subfigure}{.48\textwidth}
		\centering
		\includegraphics[width=\linewidth]{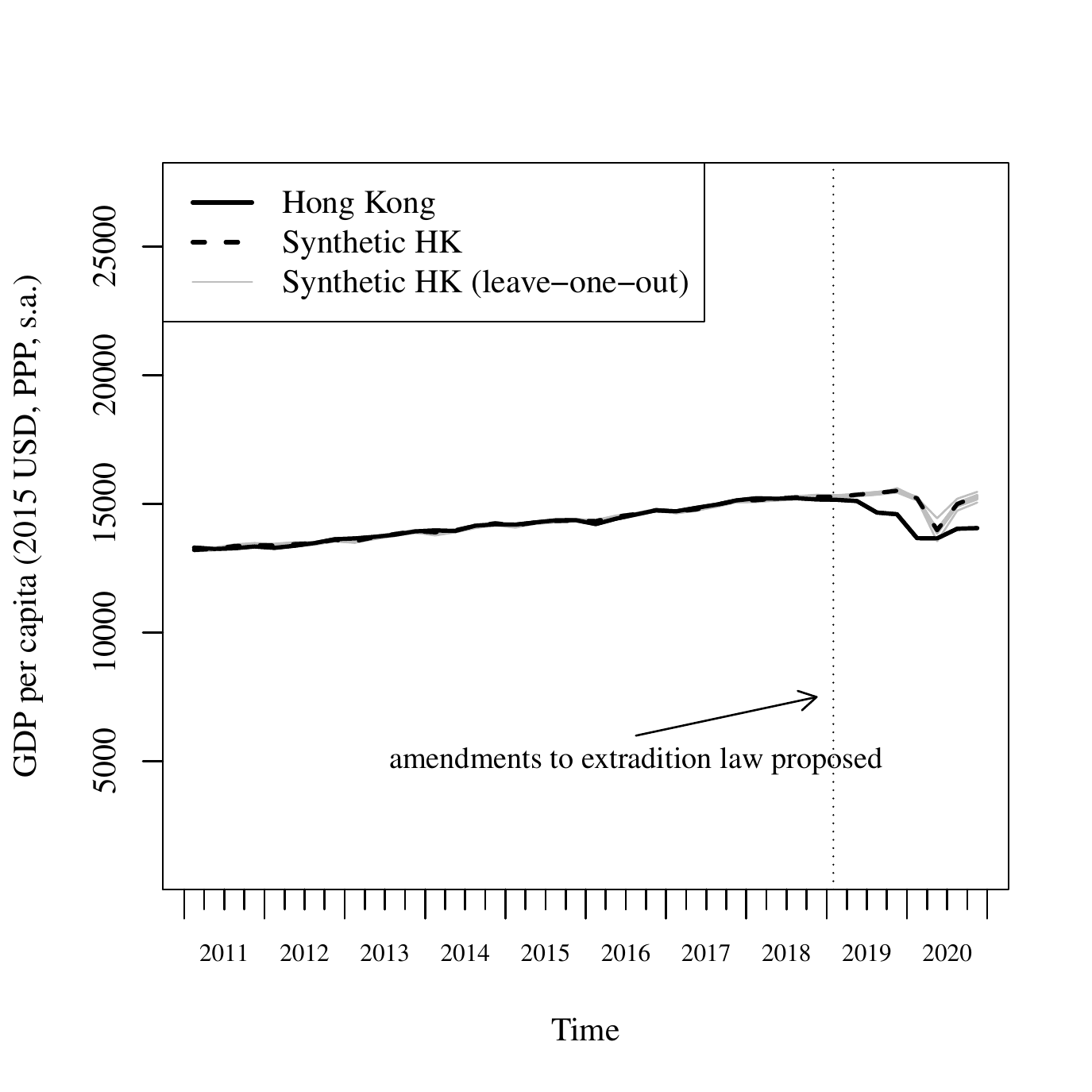}
		\caption{Leave-one-out Distribution}
		\label{Ch1_fig_HK_SC_LOO}
	\end{subfigure}
	~
	\begin{subfigure}{.48\textwidth}
		\centering
		\includegraphics[width=\linewidth]{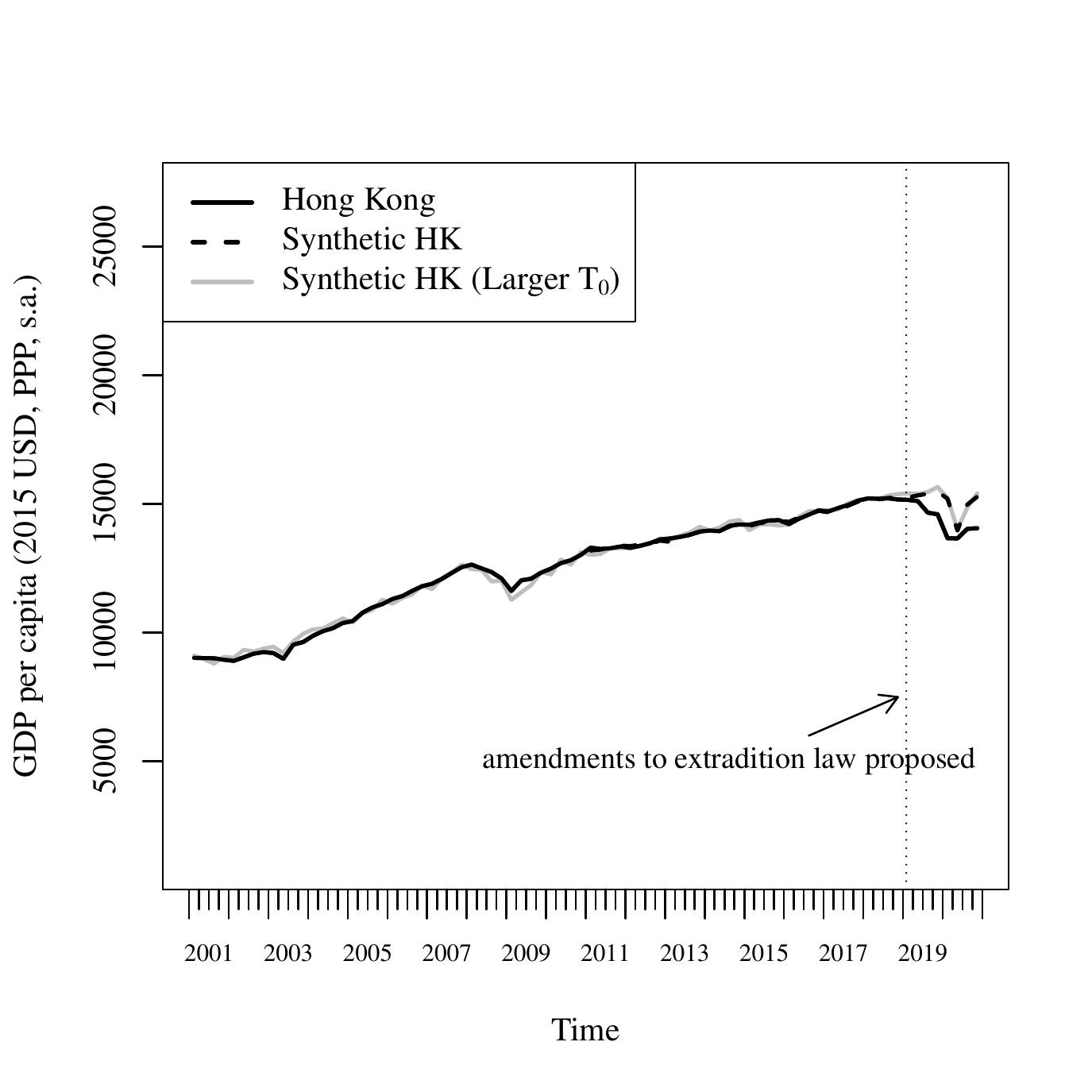}
		\caption{Longer Pretreatment Period}
		\label{Ch1_fig_HK_SC_LP}
	\end{subfigure}
	\caption{Additional Robustness Checks}
	\label{Ch1_fig_checks}
\end{figure}

Apart from the backdate exercise, we conduct two additional robustness checks.
In Figure \ref{Ch1_fig_HK_SC_LOO}, we conduct the leave-one-out exercise, where we exclude one economy at a time from the construction of the synthetic Hong Kong. We see that the trajectories for the synthetic Hong Kong constructed in the leave-one-out iterations are all very similar to the previous results. This shows that the estimated effect of the protests is robust to the exclusion of any particular economy.

We include observations from 2011 in our main analysis to avoid potential structural breaks over a longer timespan. This provides us with 36 pretreatment periods, which should be sufficient to produce credible results. Nevertheless, we check whether our results are robust to the inclusion of more pretreatment periods. In Figure \ref{Ch1_fig_HK_SC_LP}, we double the total time periods by further including outcomes observed from 2001 to 2010.\footnote{This excludes Mainland China from the analysis due to missing data.} The trajectory for the synthetic Hong Kong constructed by matching on outcomes from 2001 to 2018 follows closely the trajectory for the real Hong Kong, and the results are very similar with the benchmark results after the treatment. Thus our results are not sensitive to the inclusion of more pretreatment periods.

\section{Conclusion}\label{Ch1_sec_con}

In this paper, we generalise the synthetic control method to the case where the outcome is a nonlinear function of the underlying predictors. Specifically, we provide conditions for the asymptotic unbiasedness of the synthetic control estimator to complement the theoretical result for the linear case in \cite{abadie2010synthetic}, and propose a flexible and data-driven method for choosing the synthetic control weights.
Monte Carlo simulations show that the nonlinear synthetic control method has similar or better performance in the linear case and better performance in the nonlinear case compared with competing methods, and that the confidence intervals have good coverage probabilities across settings.
In the empirical application, we illustrate the method by estimating the impact of the 2019 anti-extradition law amendments bill protests on Hong Kong's economy, and find that the year-long protests reduced the real GDP per capita by 11.27\% in the first quarter of 2020, which is larger in magnitude than the economic decline in the 1997 Asian financial crisis and the 2008 global financial crisis.

\begin{appendices}

	\section{Data Sources}\label{Ch1_sec_data}
	\begin{itemize}
		\item Quarterly GDP for Hong Kong. Sources: Seasonally adjusted GDP in real terms and implicit price deflator (IPD) of GDP, Table E200-6, \url{https://data.gov.hk/en-data/dataset/hk-censtatd-tablechart-gdp}.
		\item Quarterly GDP for Taiwan. Sources: GDP by Expenditures - Seasonally Adjusted Series, Implicit Price Deflators, Principal Figures, \url{https://eng.stat.gov.tw/ct.asp?xItem=37408&CtNode=5347&mp=5}.
		      Purchasing Power Parity/Exchange Rate, \url{https://fred.stlouisfed.org/series/PLGDPOTWA670NRUG}.
		\item Quarterly GDP for Singapore. Sources: Gross Domestic Product In Chained (2015) Dollars, By Industry, Quarterly, Seasonally Adjusted, \url{https://www.tablebuilder.singstat.gov.sg/publicfacing/createDataTable.action?refId=16062}.
		\item Quarterly GDP for the other economies. Sources: GDP expenditure approach, Quarterly National Accounts, OECD Stat, \url{https://stats.oecd.org/}.
		\item Purchasing power parities. Sources: \url{https://data.oecd.org/conversion/purchasing-power-parities-ppp.htm}.
		\item Population. Sources: Total Population - Both Sexes, Population Dynamics, Department of Economic and Social Affairs, United Nations, \url{https://population.un.org/wpp/Download/Standard/Population/}.
	\end{itemize}

	\section{Proofs}\label{Ch1_sec_proof}

	\begin{proof}[Proof of Theorem \ref{Ch1_thm_SC}]

		The proof follows closely the proof in Appendix B of \cite{abadie2010synthetic}. We thus omit many details. For more details, see \cite{abadie2010synthetic} or \cite{botosaru2019role}.

		Under the assumptions, we have
		\begin{align*}
			e_{1t}  \equiv & Y_{1t}^{0}-\sum_j w_j^*Y_{jt}                                                                                                    \\
			        = & \boldsymbol{\lambda}_{t}'\left({\boldsymbol{\lambda}^{T_0}}'\boldsymbol{\lambda}^{T_0}\right)^{-1}{\boldsymbol{\lambda}^{T_0}}'\sum_j w_j^*\varepsilon_{j}^{T_0}                                                                \\
			        & -\boldsymbol{\lambda}_{t}'\left({\boldsymbol{\lambda}^{T_0}}'\boldsymbol{\lambda}^{T_0}\right)^{-1}{\boldsymbol{\lambda}^{T_0}}'\varepsilon_{1}^{T_0}+\varepsilon_{1t}-\sum_j w_j^*\varepsilon_{jt}. \numberthis\label{Ch1_eq_bias}
		\end{align*}

		The terms on the last line has zero conditional mean given Assumption \ref{Ch1_assume_error}, however, the term on the penultimate line does not have zero mean because $w_j^*$ is correlated with $\varepsilon_{j}^{T_0}$.

		Denote the first term as $R_{1t}$. Suppose that the elements of $|\boldsymbol{\lambda}_t|$ are bounded from above by $\bar{\boldsymbol{\lambda}}$ for $t=1,\dots,T$.
		Under Assumption \ref{Ch1_assume_rank} and using the Cauchy–Schwarz Inequality, we have
		\begin{align*}
			    & \left(\boldsymbol{\lambda}_{t}'\left(\sum_{n=1}^{T_0}\boldsymbol{\lambda}_n\boldsymbol{\lambda}_n'\right)^{-1}\boldsymbol{\lambda}_s\right) 
			\le \left(\frac{\bar{\boldsymbol{\lambda}}^2f}{T_0\underline{\xi}}\right).
		\end{align*}



		Denote $\bar{w}=\text{max}_j|w_j^*|$ ($\bar{w}\le 1$ given the adding-up and non-negativity assumptions), and $\bar{\varepsilon}_j=\sum_{s=1}^{T_0}\boldsymbol{\lambda}_{t}'\left(\sum_{n=1}^{T_0}\boldsymbol{\lambda}_n\boldsymbol{\lambda}_n'\right)^{-1}\boldsymbol{\lambda}_s\varepsilon_{js}$.
		Then using H\"{o}lder's Inequality, we have
		$$|R_{1t}| \le \bar{w}\sum_j |\bar{\varepsilon}_j| \le \bar{w}J^{1-\frac{1}{p}}\left(\sum_j |\bar{\varepsilon}_j|^p\right)^{1/p}$$
		for some positive integer $p$.


		Using H\"{o}lder's Inequality again and using Rosenthal’s Inequality, we have
		$$\mathbb{E}|\bar{\varepsilon}_j|^p \le C(p)\left(\frac{\bar{\boldsymbol{\lambda}}^2f}{T_0\underline{\xi}}\right)^p\text{max}\left\{\sum_{s=1}^{T_0}\mathbb{E}|\varepsilon_{js}|^p,\left(\sum_{s=1}^{T_0}\mathbb{E}|\varepsilon_{js}|^2\right)^{p/2}\right\},$$
		where the constant $C(p)=\mathbb{E}(\theta-1)^p$ with $\theta$ being a Poisson random variable with parameter 1.

		Denote $\bar{m}_p(T_0)=\text{max}_j (1/T_0)\sum_{s=1}^{T_0}\mathbb{E}|\varepsilon_{js}|^p$, then we have
		\begin{equation}
			\mathbb{E}|R_{1t}|< \bar{w}JC(p)^{1/p}\left(\frac{\bar{\boldsymbol{\lambda}}^2f}{\underline{\xi}}\right)\text{max}\left\{\frac{\bar{m}_p(T_0)^{1/p}}{T_0^{1-1/p}},\frac{\bar{m}_2(T_0)^{1/2}}{{T_0}^{1/2}}\right\}.\label{Ch1_eq_bound}
		\end{equation}

		Thus, the bias is bounded by a value that goes to zero when the number of pretreatment periods goes to infinity. Since $\tau_{it}=Y_{it}-Y_{it}^{0}$, this implies that $$\mathbb{E}\left(\hat{\tau}_{1t}^{SC}-\tau_{1t}\right)\rightarrow 0 \enspace \text{as} \enspace T_0\rightarrow \infty.$$

	\end{proof}

	\begin{proof}[Proof of Theorem \ref{Ch1_thm_NSC}]

		Under Assumption \ref{Ch1_assume_nY} that $G(\cdot)$ is a smooth function with $G^{(n)}<\infty$ for all $n$, we can expand $G\left(\boldsymbol{X}_j'\boldsymbol{\beta}_{t}+\boldsymbol{\mu}_j'\boldsymbol{\lambda}_{t}\right)$ for $j\in\mathcal{J}_M$ at $\boldsymbol{X}_1'\boldsymbol{\beta}_{t}+\boldsymbol{\mu}_1'\boldsymbol{\lambda}_{t}$ using Taylor's rule:
		\begin{align*}
			  & G\left(\boldsymbol{X}_j'\boldsymbol{\beta}_{t}+\boldsymbol{\mu}_j'\boldsymbol{\lambda}_{t}\right)                                                                                                                                                   \\
			= & \sum_{n=0}^\infty\frac{\partial^n G\left(\boldsymbol{X}_1'\boldsymbol{\beta}_{t}+\boldsymbol{\mu}_1'\boldsymbol{\lambda}_{t}\right)}{n!\partial \left(\boldsymbol{X}_1'\boldsymbol{\beta}_{t}+\boldsymbol{\mu}_1'\boldsymbol{\lambda}_{t}\right)^n} \\
			= & G\left(\boldsymbol{X}_1'\boldsymbol{\beta}_{t}+\boldsymbol{\mu}_1'\boldsymbol{\lambda}_{t}\right)+\sum_{n=1}^{\infty}\frac{G^{\left(n\right)}}{n!}e_{jt}^n,
		\end{align*}
		where $e_{jt}=\left(\boldsymbol{X}_j-\boldsymbol{X}_1\right)'\boldsymbol{\beta}_{t}+\left(\boldsymbol{\mu}_j-\boldsymbol{\mu}_1\right)'\boldsymbol{\lambda}_{t}$.

		Since $\sum_{j\in\mathcal{J}_M} w_j^*=1$, taking the expectation of $Y_{1t}^{0}-\sum_j w_j^*Y_{jt}$ w.r.t. $\varepsilon_{1t}$ and $\varepsilon_{jt}$, $j\in\mathcal{J}_M$, we have
		\begin{align}\label{Ch1_eq_e}
			  & \mathbb{E}_{\varepsilon}\left(Y_{1t}^0-\sum_j w_j^*Y_{jt}\right)                                                                                                                                                \\
			= & \left[\left(\sum_j w_j^*\boldsymbol{X}_j-\boldsymbol{X}_1\right)'\boldsymbol{\beta}_{t}+\left(\sum_j w_j^*\boldsymbol{\mu}_j-\boldsymbol{\mu}_1\right)'\boldsymbol{\lambda}_{t}\right]G'\nonumber \\
			  & -\sum_{n=2}^{\infty}\frac{G^{\left(n\right)}}{n!}\left(\sum_j w_j^*e_{jt}^n\right).\nonumber
		\end{align}

		If the outcomes are linear functions of the predictors, then the bias goes to zero when $T_0\rightarrow\infty$ according to Theorem \ref{Ch1_thm_SC}.
		If the outcomes are nonlinear functions, then the first term on the RHS of equation \eqref{Ch1_eq_e} may not go to zero since matching on the observed predictors and the pretreatment outcomes does not guarantee matching on the unobserved predictors, and the second term will not vanish since $\boldsymbol{H}_j-\boldsymbol{H}_1$ will not go to zero when $T_0$ and $J$ increase at the same rate.

		We now develop conditions under which the bias given by \eqref{Ch1_eq_e} converges to zero, using results from \cite{abadie2006large}.

		Equation \eqref{Ch1_eq_e} can be written as
		$\mathbb{E}_{\varepsilon}\left(Y_{1t}^0-\sum_j w_j^*Y_{jt}\right)=-\sum_{n=1}^{\infty}\frac{G^{(n)}}{n!}\left(\sum w_j^*e_{jt}^n\right).$ This bias goes to zero if $e_{jt}$ goes to zero.

		Denote $\boldsymbol{U}_j=\boldsymbol{Z}_1-\boldsymbol{Z}_j$, $j\in \mathcal{J}_M$. Without loss of generality, let $M=1+k+T_0$. Lemma 1 in \cite{abadie2006large} shows that $\boldsymbol{U}_j=O_p\left(J^{-\frac{1}{1+k+T_0}}\right)$.
		Thus for fixed $T_0$, $\boldsymbol{Z}_j-\boldsymbol{Z}_1\overset{p}{\rightarrow}0$ when $J\rightarrow\infty$.

		Since $F\left(\cdot\right)$ is a strictly monotonic function, $\boldsymbol{Z}_j-\boldsymbol{Z}_1\overset{p}{\rightarrow}0$ implies that $\boldsymbol{\lambda}^{T_0}\left(\boldsymbol{\mu}_j-\boldsymbol{\mu}_1\right)+\varepsilon_{j}^{T_0}-\varepsilon_{1}^{T_0}\overset{p}{\rightarrow}0$.
		With Assumption \ref{Ch1_assume_rank}, we have $\boldsymbol{\mu}_j-\boldsymbol{\mu}_1+\boldsymbol{\lambda}_{t}'\left({\boldsymbol{\lambda}^{T_0}}'\boldsymbol{\lambda}^{T_0}\right)^{-1}{\boldsymbol{\lambda}^{T_0}}'\left(\varepsilon_{j}^{T_0}-\varepsilon_{1}^{T_0}\right)\overset{p}{\rightarrow}0$, which implies that $\boldsymbol{\mu}_j-\boldsymbol{\mu}_1\overset{p}{\rightarrow}0$ as $T_0\rightarrow\infty$.
		However, $\boldsymbol{U}_j=O_p\left(J^{-\frac{1}{1+k+T_0}}\right)$ may not hold when $T_0\rightarrow\infty$.

		Suppose $T_0^{b(T_0)}/J=O(1)$ for some $b(T_0)\ge 1$.
		For $\boldsymbol{U}_j$ to converge to 0 when $T_0$ goes to infinity, we need $\lim_{T_0\rightarrow \infty}J^{\frac{1}{1+k+T_0}}\rightarrow \infty$, i.e., $\lim_{T_0\rightarrow \infty}T_0^{\frac{b(T_0)}{1+k+T_0}}=\lim_{T_0\rightarrow \infty}e^{\frac{b(T_0)}{1+k+T_0}\ln T_0}=\lim_{T_0\rightarrow \infty}e^{b'(T_0)\ln T_0+\frac{b(T_0)}{T_0}}\rightarrow \infty$ using L'Hôpital's rule.
		This requires $b'(T_0)>0$ since $b(T_0)\ge 1$.

		Given the regularity assumptions, $\mathbb{E}_{\varepsilon}\left(Y_{1t}^0-\sum_j w_j^*Y_{jt}\right)$ is uniformly integrable. Thus $\mathbb{E}\left(\tilde{\tau}_{1t}-\tau_{1t}\right)=\mathbb{E}\left[\mathbb{E}_{\varepsilon}\left(Y_{1t}^0-\sum_j w_j^*Y_{jt}\right)\right]\rightarrow 0$ when $T_0\rightarrow \infty$ if
		$J=O\left(T_0^{b(T_0)}\right)$ with $b(T_0)\ge 1$ and $b'(T_0)>0$.

	\end{proof}

\end{appendices}

 \bibliography{References1}
 \bibliographystyle{apalike}

 \end{document}